\documentclass{article}

\usepackage[english]{babel}

\usepackage[a4paper,top=2cm,bottom=2cm,left=3cm,right=3cm,marginparwidth=1.75cm]{geometry}

\usepackage{amsmath}
\usepackage{amssymb}
\usepackage{amsthm}
\usepackage{bm}
\usepackage{graphicx}
\usepackage[colorlinks=true, allcolors=blue]{hyperref}
\usepackage{xcolor}
\usepackage{subcaption}
\usepackage{comment}

\usepackage{cite}
\usepackage{authblk}

\newtheorem{thm}{Theorem}
\newtheorem{ansatz}{Conjecture}
\newtheorem{lemma}{Lemma}
\newtheorem{proposition}{Proposition}

\theoremstyle{remark}

\newcommand{\meanv}[1]{\left\langle#1\right\rangle}
\newcommand{\bxi}{\boldsymbol{\xi}}
\newcommand{\bs}{\boldsymbol{s}}

\newcommand{\bxip}{\boldsymbol{\hat{\xi}}}
\newcommand{\bS}{\bm{\mathcal{S}}}

\DeclareMathOperator{\extr}{\mathrm{Extr}}  

\title{Hopfield model with planted  patterns: a teacher-student self-supervised learning model}

\author[1]{Francesco Alemanno}

\author[1]{ Luca Camanzi}

\author[1]{Gianluca Manzan}

\author[1]{Daniele Tantari \thanks{Corresponding Author: daniele.tantari@unibo.it}}

\affil[1]{Department of Mathematics, University of Bologna, 

Piazza di Porta San Donato 5, 40126, Bologna (BO), Italy}

\begin{document}

\maketitle

\begin{abstract}
While Hopfield networks are known as paradigmatic models for \textit{memory storage} and retrieval, modern artificial intelligence systems mainly stand on the machine \textit{learning} paradigm. We show that it is possible to formulate a teacher-student self-supervised learning problem with Boltzmann machines in terms of a suitable generalization of the Hopfield model with structured patterns, where the spin variables are the machine weights and patterns correspond to the training set's examples. We analyze the learning performance by studying the phase diagram in terms of the training set size, the dataset noise and the inference temperature (i.e. the weight regularization). With a small but informative dataset the machine can learn by memorization. With a noisy dataset, an extensive number of examples above a critical threshold is needed. In this regime the memory storage limits  becomes an opportunity for the occurrence of a learning regime in which the system can generalize.
\end{abstract}

\section{Introduction}


What is the maximum amount of patterns that can be stored by a neural network and efficiently retrieved?
What is the minimum amount of examples needed for a neural network to understand the hidden structure of a noisy, high dimensional dataset? 
They seem two different questions, the former concerning the limit of a memorization mechanism, the latter involving the beginning of a learning process. In facts they are strictly related, being a common life experience, especially for science students, that learning comes to help when memory starts to fail.

In this work we investigate this relation by introducing a suitable generalization of the Hopfield model that naturally emerges in the context of a teacher-student, self-supervised learning problem.

The Hopfield model \cite{hopfield1982neural} is the statistical mechanics paradigm for associative memory and describes a complex system of connected neurons able to retrieve patterns of information previously stored in their interactions through the so called Hebb rule \cite{hebb2005organization}. In the case of uncorrelated patterns, retrieval (thus memorization) is possible up to a critical load, that scale linearly with the system's size, beyond which a spin-glass regime occurs where the system gets confused \cite{amit1987statistical,coolen2005theory}. Since the Hopfield seminal work, several generalizations have been investigated in relation to their critical storage capacity and retrieval capabilities. For example, super-linear capacity has been found by allowing multi-body interactions \cite{gardner1987multiconnected}, parallel retrieval has been studied in relation to patterns sparsity \cite{agliari2012multitasking,agliari2013immune,agliari2013immunemedium,agliari2017retrieving,sollich2014extensive} or hierarchical interactions \cite{agliari2015retrieval,agliari2015hierarchical,agliari2014metastable,agliari2015topological} and non-universality has been shown with respect to more general patterns entries and unit priors \cite{barra2018phase,barra2012glassy,agliari2017neural,genovese2020legendre,agliari2015anergy,rocchi2017high}.  A specific attention has been given to model with intra  pattern and among patterns correlation \cite{fontanari1990storage,der1992modified,van1997hebbian,agliari2022storing}. In this works the negative effects of the  correlation structure on the system's capacity emerge and different non-Hebb rules are proposed to mitigate it, restoring the possibility of pattern retrieval.  Other works on the effect of patterns correlation on the critical capacity are \cite{lowe1998storage,gutfreund1988neural,de2023effect}.

From this perspective it appears as a collective endeavour in favour of an artificial intelligence (AI) that works exclusively by machine memorization, in contrast to modern AI that is being dominated by machine learning. Recent models have started to consider patterns correlation as connected to the existence of a new regime which is more close to learning than memory retrieval. For example in \cite{agliari2022emergence,alemanno2022supervised,huang2016unsupervised,huang2017statistical} patterns are blurred copies of some prototypes while in  \cite{mezard2017mean,ichikawa2022statistical,negri2023hidden} patterns are generated from a set of hidden features: in both cases there exist a regime where the Hopfield network can extract (learn) the underlying structure more than merely memorize the examples.

In this work we show that this type of mechanism is exactly what happens in a problem of self-supervised learning \cite{liu2021self}, where a machine is trained to learn the probability distribution of a given dataset from a subset of unlabelled examples\footnote{In this paper the term self-supervised learning is considered as synonymous of the more classical \textit{unsupervised} learning and in contrast with supervised learning from labelled data.}. Depending on the  noise of the dataset and  weight regularization the machine stops working by memorization and starts learning by generalization.

The paper is organized as follows. In the Introduction a generalization of the Hopfield model with planted correlated patterns is presented together with its natural interpretation in terms of the machine weights posterior distribution in a teacher-student self-supervised learning problem. In the Section Results the model is studied in three different regimes: in the Bayes optimal \cite{iba1999nishimori,barbier2019adaptive} setting the machine can work by memorization if the dataset noise is relatively small while it can retrieve the signal by generalization when the training set is made of a sufficiently high number of weakly informative examples; beyond Bayes-optimality \cite{angelini2023limits,angelini2021mismatching} the learning regime depends on the relation between dataset noise and weight regularization (inference temperature). In the last Section the proofs of the main results are given while in the Appendix the derivation of the Conjectures is provided.


 We consider a Hopfield model with $N$ binary spins $\bxi=(\xi_1,\ldots , \xi_N)\in \{-1,1\}^N$ and $M$ quenched random binary patterns ${\bm{\mathcal{S}}:=\{\bs^\mu\}^M_{\mu=1}}=\{s^\mu_1,\ldots,s^\mu_N \}^M_{\mu=1}\in \{-1,1\}^{NM}$. Given a specific realization of the patterns and a planted configuration $\bxip\in \{-1,1\}^N$, we consider the Boltzmann-Gibbs distribution
\begin{equation}\label{eq:hopfield}
    \hat{P}(\bxi|\bS,\bxip)= \hat{Z}^{-1}(\bS,\bxip) \exp\left( \frac{\beta}{N}{\sum_{\mu=1}^M}\sum_{i<j}^N s_i^\mu s_j^\mu\xi_i\xi_j+\lambda \sum_{i=1}^N \hat{\xi}_i \xi_i\right),
\end{equation}
where $\beta\geq 0$ is the inverse temperature, $\lambda\geq0$ is the amplitude of an external field in the direction of the planted configuration and 
\begin{equation}
   Z(\bS,\bxip)=\sum_{\bxi}  \exp\left( \frac{\beta}{N}{\sum_{\mu=1}^M}\sum_{i<j}^N s_i^\mu s_j^\mu\xi_i\xi_j+\lambda \sum_{i=1}^N \hat{\xi}_i \xi_i\right)
\end{equation}
is the model partition function. Differently from the standard Hopfield model, where patterns consist of i.i.d. Rademacher random variables, we assume the patterns are independent but with a specific spatial correlation induced by the planted configuration  i.e. they are distributed according to 
\begin{equation}\label{eq:datadistribution}
    P(\bS|\bxip)=\prod_{\mu=1}^M P(\bs^\mu|\bxip)= \prod_{\mu=1}^M z_{\mu}^{-1}\exp\left(\frac{\beta}{N}\sum_{i<j}^N\hat{\xi_i}\hat{\xi_j}s^\mu_is^\mu_j+h^\mu\sum_{i=1}^N \hat{\xi}_is^\mu_i\right),
\end{equation}
 where the partition function 
\begin{equation}
    z_{\mu} :=  \sum_{\bm{s}}\exp\left(\frac{\beta}{N}\sum_{i<j}^N\hat{\xi_i}\hat{\xi_j}s_is_j+h^\mu\sum_{i=1}^N \hat{\xi}_is^\mu_i\right)=
    \sum_{\bs}\exp\left(\frac{\beta}{N}\sum_{i<j}^Ns_is_j+h^\mu\sum_{i=1}^Ns^\mu_i\right) ,
\end{equation}
is nothing but the partition function of the classical Curie-Weiss model  at  inverse temperature $\beta$ and external field $h^{\mu}\in\mathbb{R}$. In the second equality we just use the Gauge transformation $s_i\to s_i\hat{\xi}_i$. Finally we assume the planted configuration $\bxip$ is a quenched Rademacher random vector.

The model (\ref{eq:hopfield}, \ref{eq:datadistribution}), in particular in the limit $\lambda,\bm{h}\to 0$, is motivated by its natural application  in the context of self-supervised learning, where a machine is trained over a dataset (training set) of unlabelled examples to retrieve their probability distribution. To this task the Hopfield model belongs to the class of the Restricted Boltzmann Machines (RBMs) \cite{ackley1985learning},  that are widely used architectures and consist of neural networks with two layers of units. The first (visible) layer reproduces the data with its units $\bs=(s_1,\ldots, s_N)$, the second (hidden) layer has the role of building an internal representation of the data structure using its units $\bm{\tau}=(\tau_1,\ldots, \tau_P)$. In the statistical mechanics literature they are also studied in the context of multi-species spin-glass models \cite{barra2015multi,panchenko2015free,alberici2021multi,alberici2020annealing,alberici2021deep,genovese2022remark,genovese2016non,genovese2015legendre,barra2014mean,genovese2017overlap,barra2014solvable}. In practice a RBM is the parametric probability distribution 
\begin{equation}
    P_{\bm{w}}(\bs)=z_{\bm{w}}^{-1}\mathbb{E}_{\bm{\tau}}\exp\left( \sum_{i,j}^{N,P} w^j_i s_i\tau_j \right),
\end{equation}
whose parameters $\bm{w}\in \mathbb{R}^{NP}$ need to be fit with the data. The expectation $\mathbb{E}_{\bm{\tau}}$ is intended over the a priori distribution of the hidden units. It is well known \cite{barra2012equivalence,agliari2018non} that if one assumes gaussian units $\bm{\tau}\sim\mathcal{N}(0,\mathbb{I})$ then the RBM is exactly a Hopfield model with $P$ patterns $\{\bm{w}^{\mu}\}_{\mu=1}^P$. In a self-supervised learning setting, given an unknown probability distribution $P^0$ and a training set of $M$ examples ${\bm{\mathcal{S}}=\{\bs^\mu\sim P^0\}^M_{\mu=1}}$ drawn independently from $P^0$, the aim is to approximately learn $P^0$ with $P_{\bm{w}}$ by tuning $\bm{w}$. The learning performance clearly depends on the structure of the data, i.e. $P^0$, the properties of the machine, i.e. $P_{\bm{w}}$, and the amount of data.
A crucial research question is thus about the typical size of the training set necessary for the machine to efficiently learn, given its architecture and the structure of the data.  To answer this question we can consider a controlled scenario in which the dataset $\bS$ is generated from a (teacher) machine $P_{\bm{\hat{w}}}$ and another (student) machine $P_{\bm{w}}$ is trained over $\bS$.  In the case $P=1$, choosing $\bm{\hat{w}}=\sqrt{\beta/N}\bxip$, we find that the probability distribution of the dataset is exactly as in Eq.\ref{eq:datadistribution}, i.e. it has a spatial correlation induced by the planted configuration $\bxip$. The parameter $\beta^{-1}$ can be interpreted as the amount of noise in the training set examples or at the same time the typical strength of the machine weights, i.e. weights regularization.  In a Bayesian framework, the posterior distribution of the student's machine weights $\bm{w}=\sqrt{\beta/N}\bxi$ given the dataset reads as
\begin{equation}\label{eq:post_prob}
    \hat{P}(\bm{\xi}|\bm{\mathcal{S}})=\frac{P(\bxi)\prod^M_{\mu=1}P(\bm{s}^\mu|\bm{\xi})}{ P(\bm{\mathcal{S}})} = Z^{-1}(\bm{\mathcal{S}})\exp\left(\frac{\beta}{N}{\sum_{\mu=1}^M}\sum_{i<j}^N s_i^\mu s_j^\mu\xi_i\xi_j\right),
\end{equation}
which is exactly the Hopfield model of Eq. \ref{eq:hopfield} in absence of fields. In the following we refer to the machine weights $\bxi$ as the student pattern to distinguish them from the teacher (planted) pattern $\bxip$, also denoted as the \textit{signal}. In a statistical inference framework Eq. \ref{eq:datadistribution} define the so called \textit{direct} Hopfield model describing the dataset, while Eq. \ref{eq:post_prob}  can be considered as the \textit{inverse} model. Interestingly it is still a Hopfield model, \textit{dual} w.r.t. the direct model,  where the spin variables correspond to the machine weights $\bxi$ while the training set’s examples $\bS$ play the role of the patterns.
We also refer to the dual patterns as planted disorder to enlighten that they are in turn drawn from a Hopfield model with a planted pattern, i.e. $\bxip$. This teacher-student setting has been introduced in \cite{barra2017phase}, originally inspired by \cite{cocco2011high} and recently studied in the case $P=2$ for boolean RBM \cite{hou2019minimal}. It was also used in \cite{decelle2021inverse} to propose alternative posterior based methods for inverse problems in the case of structured dataset.

In this setting it is possible to quantify the learning performance of the student machine by measuring how much the student pattern $\bxi$ is close to the signal $\bxip$, i.e. by computing $Q(\bxi,\bxip)$, once introduced the overlap between two vectors $\bxi^1,\bxi^2\in\mathbb{R}^N$ as 
\begin{equation}
    Q(\bxi^1,\bxi^2)=\frac{\bxi^1\cdot\bxi^2}{N}.
\end{equation}
Similarly we can evaluate the amount of information contained in the data with the overlap $Q(\bs^\mu,\bxip)$ between the examples and the teacher pattern. Finally we can characterize the memorization performance of the machine by introducing the overlap $Q(\bs^\mu,\bxi)$ between the examples and the student pattern (sampled from the posterior (\ref{eq:post_prob})). 
 We define the bracket $\meanv{.}^{\hat{}}$ as the expected value w.r.t the joint distribution of  signal,  training set and student pattern $(\bxip,\bS,\bxi)$, i.e. for any function $f:\{-1,1\}^N\times\{-1,1\}^{NM}\times\{-1,1\}^N\to\mathbb{R}$,
\begin{equation}
    \meanv{f}^{\hat{}}:=
2^{-N}\sum_{\bxip,\bS,\bxi} f(\bxip,\bS,\bxi)\hat{P}(\bm{\xi}|\bm{\mathcal{S}},\bxip) P(\bS|\bm{\hat{\xi}}).
\end{equation}
It is interesting to note that using the Gauge transformation $\xi_i \to \xi_i\hat{\xi}_i$ and $s^\mu_i\to s^\mu_i \hat{\xi}_i$,   we get for any bounded function $f$ of the overlap that
\begin{eqnarray}
  \meanv{f(Q(\bxi,\bxip))}^{\hat{}} &=&2^{-N}\sum_{\bxip,\bS,\bxi} f(Q(\bxi,\bxip))\hat{P}(\bm{\xi}|\bm{\mathcal{S}},\bxip) P(\bS|\bm{\hat{\xi}})  \nonumber\\  
  &=&\sum_{\bS,\bxi} f(Q(\bxi,\bm{1})) \hat{P}(\bm{\xi}|\bm{\mathcal{S}},\bm{1})P(\bS|\bm{1})= \meanv{f(Q(\bxi,\bm{1}))} \label{eq:gauge},
\end{eqnarray}
where we have defined the bracket $\meanv{.}$ as the expectation w.r.t. the joint distribution $\hat{P}(\bm{\xi}|\bm{\mathcal{S}},\bm{1})P(\bS|\bm{1})$ that does not depend on the signal $\bxip$ (ferromagnetic gauge), i.e. for any function $g:\{-1,1\}^{NM}\times\{-1,1\}^N\to\mathbb{R}$,
\begin{equation}
    \meanv{g}:=
\sum_{\bS,\bxi} g(\bS,\bxi)\hat{P}(\bm{\xi}|\bm{\mathcal{S}},\bm{1})P(\bS|\bm{1}).
\end{equation}
Note that $\hat{P}(\bm{\xi}|\bm{\mathcal{S}},\bm{1})$ is the Boltzmann-Gibbs distribution induced by the partition function
\begin{equation}\label{eq:z_gauge}
   Z(\bS)= \sum_{\bxi}\exp\left(\frac \beta N \sum_{\mu=1}^M\sum_{i<j}^N s^\mu_is^\mu_j \xi_i \xi_j +\lambda \sum_{i=1}^N \xi_i\right),
\end{equation}
that is a Hopfield model in a field, whose patterns are drawn from 
\begin{equation}\label{eq:pattern_d}
   P(\bS|\bm{1})= \prod_{\mu=1}^M z^{-1}_\mu\exp \left(\frac{\beta}{N}\sum_{i<j}s_i^\mu s_j^\mu + h^\mu\sum_{i=1}^N s^\mu_i \right),
\end{equation}
i.e. the distribution of $M$ independent Curie-Weiss models at inverse temperature $\beta$ and external field $\bm{h}$. In the following we indicate with $\mathbb{E}_{\beta,\bm{h}}$ the expectation w.r.t. the pattern distribution (\ref{eq:pattern_d}). Eq. $(\ref{eq:gauge})$ means that the overlap with the signal can be interpreted as the magnetization of a Hopfield model at inverse temperature $\beta$ with patterns $\bS$ extracted independently from a Curie-Weiss model at the same temperature. 
Analogously it holds 
\begin{equation}
  \meanv{f(Q(\bs^\mu,\bxi))}^{\hat{}} = \meanv{f(Q(\bs^\mu,\bxi))}
\end{equation}
\begin{equation}
\meanv{f(Q(\bs^\mu,\bxip))}^{\hat{}}=\meanv{f(Q(\bs^\mu,\bm{1}))}
\end{equation}
i.e. the overlap between the examples and the signal corresponds, in the ferromagnetic gauge, to the magnetization of $\bS$.
For this reason in the following we always consider, without loss of generality, that the patterns are drawn from $(\ref{eq:pattern_d})$, which does not depend on the signal $\bxip$, keeping in mind that a non-zero magnetization in this model corresponds to a macroscopic alignment with the planted configuration $\bxip$.

\section{Results}

\subsection{Learning by memorization from few highly informative examples}

At low enough temperature we expect that the examples $\bS$ are polarized, i.e. in terms of the original variables they are correlated with the planted configuration $\bxip$. At the same time, as long as the number of examples do not exceed the critical load of the student machine, the student pattern $\bxi$ will be aligned with one of them and therefore we expect it to be polarized as well. From the machine learning perspective, since each example  carries a lot of information about the signal, the student machine can easily learn $\bxip$ by memorization, even if $M=1$.


We can formalize this result and precisely quantify the goodness of the learning performance in terms of the temperature $\beta$ and the amount of data $M$ by computing the system's free energy
\begin{equation}
    f_N = -\frac 1 {\beta N}  \mathbb{E}_{\beta,\bm{h}}\log Z(\bS)
\end{equation}
where $Z(\bS)$ is  given by Eq. (\ref{eq:z_gauge}).
 In the thermodynamic limit we have the following 
\begin{thm}\label{thm:1}
If $\beta\leq 1$, or $\beta>1$ and $\lambda,h\in\mathbb{R}\setminus\{0\}$,
it holds for any $\bm{\epsilon}\in \{-1,1\}^M$
\begin{eqnarray}
-\beta f &:=&  \lim_{N\to\infty} \frac 1 N \mathbb{E}_{\beta,h\bm{\epsilon}}\log Z(\bS) 
= \operatorname{sup}_{\bm{p}\in\mathbb{R}^M} g(\bm{p})
\end{eqnarray}
with
\begin{equation}
g(\bm{p})=\log 2 -\frac{\beta \bm{p}^2}{2} +\meanv{\log\cosh \left( \beta \bs\cdot\bm{p}  +\lambda\right)}_{\bs},
\end{equation}
where we have defined the random vector $\bs\in \{-1,1\}^M$ whose entries are i.i.d. random variables with mean \begin{equation}
m_0(\beta,h):=\operatorname{argmax}_{x\in\mathbb{R}} \left[\log 2 +\log\cosh(\beta x + h)- \frac{\beta x^2}{2}\right].
\end{equation}
\end{thm}
The variational principle expressing the free energy corresponds to what one would expect for a Hopfield model at a low load of biased patterns and this is natural since the model for the dual patterns (\ref{eq:datadistribution}), as well as the Curie Weiss model (\ref{eq:pattern_d}), is mean-field in the thermodynamic limit, therefore spatial correlations vanish.  The solution of the variational principle is a stationary point of $g(\bm{p})$, therefore a solution of
\begin{equation}\label{eq:sc0}
    \bm{p} 
    = \left\langle \bm{s} \tanh(\beta \bm{s} \cdot \bm{p} +\lambda)\right\rangle_{\bm{s}} 
    =
    \frac{\sum_{\bm{s}}  \bm{s}\  e^{\beta m_0\bm{1}\cdot \bm{s}} \tanh(\beta \bm{s}\cdot\bm{p}+\lambda)}{\left(2 \cosh(\beta m_0)\right)^M}.
\end{equation}
 Note that Theorem \ref{thm:1} states that the limiting free energy doesn't depend on $\bm{\epsilon}$, therefore we can consider without loss of generality the case of a positive uniform external field $\bm{h}=h \bm{\epsilon}=h\bm{\bm{1}}$ acting on the examples. In terms of the  learning scenario this corresponds to saying  that the student does not care whether the examples are aligned or anti-aligned with the planted configuration $\bxip$. Note that this is not  a symmetry by global spin flipping because each example field  can have a different sign and thus a different alignment w.r.t $\bxip$. This result is particularly useful because in the limit of zero external field it is well known that the Curie Weiss measure at low temperature  is a mixture measure $P^{CW}=1/2(P^{m_0}+P^{-m_{0}})$ and consequently the training set $\bS$ is in general composed of two clusters of examples with opposite global magnetization: this is in general a complication when dealing with inverse problems \cite{braunstein2011inference,decelle2016solving,decelle2021inverse}. Nevertheless, using a Bayesian approach and thanks to the resulting  Hebbian interaction, the posterior distribution works exactly as the examples were all aligned in the same direction. 

From the solution of the free energy variational principle  all the model order parameters can be derived according to the following
\begin{proposition}\label{prop:1}
Assuming $\bm{\epsilon}=\bm{1}$ and given $\bm{p}\in\mathbb{R}^M$ the global maximizer of $g(\bm{p})$, then it holds
\begin{align}
&\bullet\lim_{N\to\infty}\meanv{Q(\bs^\mu,\bxip)}^{\hat{}}=\lim_{N\to\infty} \meanv{\frac{\bs^\mu\cdot\bm{1}}{N}}= m_0(\beta,h)&& \forall \mu=1,\ldots,M;\\
&\bullet\lim_{N\to\infty}\meanv{Q(\bs^\mu,\bxi)}^{\hat{}}=\lim_{N\to\infty} \meanv{\frac{\bs^\mu\cdot\bxi}{N}}=p^\mu&& \forall \mu=1,\ldots,M;\\
&\bullet\lim_{N\to\infty}\meanv{Q(\bxip,\bxi)}^{\hat{}}=\lim_{N\to\infty} \meanv{\frac{\bm{1}\cdot\bxi}{N}}=m=\left\langle  \tanh(\beta \bm{s} \cdot \bm{p}+\lambda )\right\rangle_{\bm{s}}.&&
\end{align}
Moreover the random variables $Q(\bs^\mu,\bxip)$, $Q(\bs^\mu,\bxi)$ and $Q(\bxip,\bxi)$ are self-averaging.
\end{proposition}

\noindent Since we are interested in the limit $h,\lambda\to 0$, we need to study the system of equations
\begin{equation}\label{eq:selfcon_p}
    \bm{p} 
    = \left\langle \bm{s} \tanh(\beta \bm{s} \cdot \bm{p} )\right\rangle_{\bm{s}} 
    =
    \frac{\sum_{\bm{s}}  \bm{s}\  e^{\beta m_0(\beta)\bm{1}\cdot \bm{s}} \tanh(\beta \bm{s}\cdot\bm{p})}{\left(2 \cosh(\beta m_0(\beta))\right)^M},
\end{equation}
where $m_0(\beta)=m_0(\beta,0^+)$.
In the case $M=1$ the equation (\ref{eq:selfcon_p}) takes the form
\begin{equation}\label{eq:p_M=1}
    p= \left\langle s \tanh(\beta s p) \right\rangle_{s} = \tanh(\beta p).
\end{equation}
whose solutions are $p=\pm m_0(\beta)$ and from which 
\begin{align}
   m=  \left\langle \tanh(\beta s p ) \right\rangle_{s} =\left\langle s \right\rangle_{s} \tanh(\beta p )=  m_0 \tanh(\beta p)=\pm m_0(\beta)^2,
\end{align}
A similar ferromagnetic bifurcation occurs also for generic number of examples $M$. In fact  when $\beta\leq 1$ the only solution for the example magnetization is $m_0=0$, therefore the average $\left\langle \cdots\right\rangle_{\bm{s}}$ becomes  uniform  over  $\{-1,1\}^M$. As a consequence, using $|\tanh(z)|<|z|$, it holds
\begin{align}
    \bm{p}^2 &=\left\langle(\bm{s}\cdot\bm{p})\tanh(\beta\bm{s}\cdot\bm{p})\right\rangle_{\bm{s}}\leq\left\langle|\bm{s}\cdot\bm{p}||\tanh(\beta\bm{s}\cdot\bm{p})|\right\rangle_{\bm{s}}\nonumber\\
    & \leq \beta \left\langle(\bm{s}\cdot\bm{p})^2\right\rangle_{\bm{s}}=\beta \sum_{\mu,\nu} p_\mu p_\nu \left\langle s_\mu s_\nu \right\rangle_{\bm{s}}=\beta \bm{p}^2.
    \label{eq:paraNishi}
\end{align}
Hence, if $\beta<1$ the only solution is $\bm{p}=0$, from which $m=0$. This means that at high temperature there is no information about the planted pattern in a dataset composed of a finite number of examples, simply because they are uncorrelated with the signal $\bxip$.
It is immediate to verify that $\bm{p}=\bm{0}$ is a solution at any temperature, however, when $\beta>1$ it is unstable and other solutions with non zero overlap $\bm{p}$ appear. It is possible to characterize those solutions thanks to the following propositions.
\begin{proposition}\label{prop:p}
The solutions of Eqs. (\ref{eq:selfcon_p}) have equal components, i.e. $p^\mu=\bar{p}$ $\forall \mu=1,\ldots, M$.
\end{proposition}
\noindent Thanks to Proposition \ref{prop:p} it is sufficient to solve the one-dimensional equation
\begin{equation}\label{eq:new_p}
    \overline{p}=\left\langle s_1 \tanh(\beta \overline{p} \sum_{\mu=1}^M s_\mu)\right\rangle_{\bm{s}}=\frac{\sum_{\bm{s}}s_1\ e^{\beta m_0 \sum_{\mu=1}^M s_\mu } \tanh(\beta \overline{p} \sum_{\mu=1}^M s_\mu)}{\left(2\cosh(\beta m_0)\right)^M},
\end{equation}
where $\overline{p}$ is the value of each component of $\bm{p}$. By studying Eq. (\ref{eq:new_p}) we have the following
\begin{proposition}\label{prop:2}
The points $p^\mu=\bar{p}=\pm m_0(\beta)$, $\forall \mu=1,\ldots, M$, are solutions of eqs. (\ref{eq:selfcon_p}).
\end{proposition}
\noindent To check whether they are actually minimizers one should look at the free energy, obtaining the following
\begin{proposition}\label{prop:3}
For $\beta>1$, $\lambda=0$ and $h=0^+$, the maximum of $g$ is attained in the two symmetric points $p^\mu=\pm m_0(\beta)$, $\forall \mu=1,\ldots, M$.
\end{proposition}
\noindent The previous results show that when the temperature is low enough, i.e. $\beta>1$, both the example magnetization $m_0$ and the overlap of the system with the examples $\bar{p}$ are different from zero. This means that the examples are all macroscopically aligned with the signal  $\bxip$ (i.e. they are largely informative) and that $\bxi$ is macroscopically aligned with all the examples. As a consequence the system must be macroscopically aligned with the signal (learning is possible and easy) and in fact in this regime $m\neq 0$.  It is interesting to note that while $\bar{p}$ does not depend on $M$, i.e. the alignment with the examples only depends from the posterior temperature, the system magnetization $m$, i.e. the alignment with the signal, increases with the size of the dataset. In fact its value is simply obtained as
\begin{equation}
    m=\meanv{\tanh(\beta \bs\cdot\bm{p})}_{\bm{s}}
    =\frac{\sum_{\bm{s}}\ e^{\beta M m_0 m(\bs) }\tanh(\beta M \overline{p} m(\bs) )}{\left(2\cosh(\beta m_0)\right)^M },
\end{equation}
where $m(\bs)=M^{-1}\sum_{\mu=1}^M s^\mu$.
We report in Figure \ref{fig:comparison_mag} the value of the magnetization as a function of the temperature for different size of the dataset $M$.
\begin{figure}[ht]
\centering
\includegraphics[width=0.47\columnwidth]{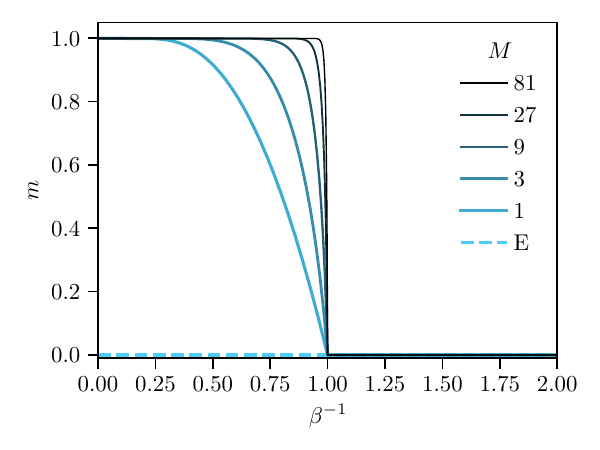}
\includegraphics[width=0.47\columnwidth]{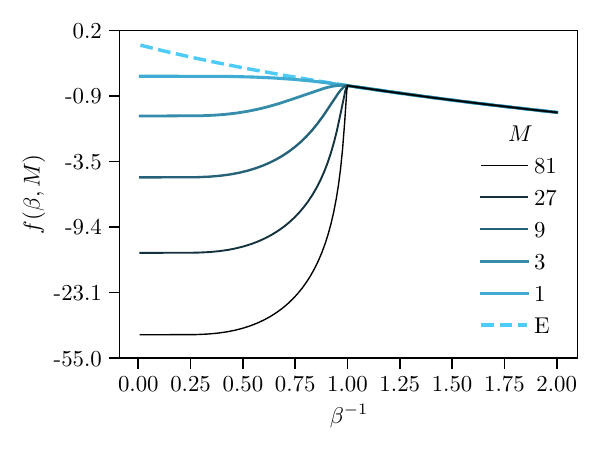}
\caption{
Learning performance with a finite $M$ dataset size. \textit{Left:} System's magnetization, i.e. the overlap between teacher and student pattern is evaluated as a function of the temperature $\beta^{-1}$. The overlap increases with the number of examples $M$ as long as the system is below the critical temperature.  \textit{Right:} System's free energy as a function of $\beta^{-1}$. The free energy corresponding to solutions with $m>0$ are painted in solid lines, while the ergodic (E) free energy, i.e. the one corresponding to $m=0$, appears with a dashed line. As long as $\beta^{-1}<1$, the global minimum of the free energy is the state where the machine can learn the original pattern.
}
\label{fig:comparison_mag}
\end{figure}
\noindent As $M$ increases, it is evident that the magnetization tends to 1 for $\beta>1$. The system's free  energy
\begin{equation}
    f=\frac{M}{2} \overline{p}^2 - \frac{1}{\beta} \frac{\sum_{\bm{s}}e^{\beta M m_0 m(\bs)}\ln 2\cosh(\beta M\overline{p} m(\bs))}{\left(2\cosh(\beta m_0)\right)^M}
\end{equation}
is also displayed to show 
the instability of the solution $\overline{p}=0$ at low temperature.



\subsection{Learning by generalization from many noisy examples}

In the previous Section we have shown that when the examples are highly correlated with the signal, i.e. at $\beta>1$, learning is possible and easy, with any finite number of examples $M>0$. Conversely when the examples are poorly correlated with the signal, i.e. $\beta <1$, there is not enough information in the posterior distribution to retrieve the original pattern. In this Section we show that learning is possible also in the low correlation regime (high temperature) as long as we consider a larger dataset. In particular we consider the case in which the machine can leverage on an extensive number of examples, i.e.
\begin{equation}
    \lim_{N\to\infty} M/N = \gamma >0.
\end{equation}
In this regime the free energy can only be derived exploiting the replica method under the replica symmetric approximation \cite{nishimori2001statistical}
from which one gets the following
\begin{ansatz}\label{conj:RS}
For $\lambda=0$, $\bm{h}=0$, $\beta<1$ and $\gamma>0$, the limiting  free energy of the posterior distribution is
\begin{equation}
    -\beta f = \lim_{N\to\infty}\frac 1 N \mathbb{E}_{\beta,\bm{h}}Z(\bS)=\operatorname{Extr}_{m,\hat{m},q,\hat{q}} f (m,\hat{m},q,\hat{q}),
\end{equation}
where 
\begin{align}\label{eq:f}
    f (m,\hat{m},q,\hat{q}) = &-\frac{\gamma}{2} \left[ \ln\big((1-\beta)(1-\beta+\beta q)\big) -\frac{\beta q }{(1-\beta)}+\frac{ \beta^2 (q^2-m^2)}{(1-\beta)(1-\beta+\beta q)}\right]\nonumber \\
    &+ \frac{\hat{q}q}{2} - \hat{m}m - \frac{\hat{q}}{2} + \int D\mu(z)\ln 2\cosh(\hat{m}+z\sqrt{\hat{q}}). 
\end{align}
Thus the saddle point equations read as
\begin{align}
    \label{eq:m}m&=\int D\mu(z) \tanh(\hat{m}+z\sqrt{\hat{q}})\\ \label{eq:mh}\hat{m}&= \frac{\gamma \beta^2 m}{(1-\beta)(1-\beta+\beta q)}\\
    \label{eq:q}q &= \int D\mu(z) \tanh^2(\hat{m}+z\sqrt{\hat{q}})\\
    \label{eq:qh}\hat{q} &=  \frac{\gamma \beta^3(m^2-q^2)}{(1-\beta)(1-\beta+\beta q)^2}+\frac{\gamma \beta^2 q}{(1-\beta)(1-\beta+\beta q)}.
\end{align}
\end{ansatz}
\noindent The solution of the saddle point equations have a physical interpretation in terms of the model's order parameters according to the following
\begin{ansatz}\label{conj:RSparameters}
Given $(m,q)$ solutions of Eqs. (\ref{eq:m},\ref{eq:q}) it holds
\begin{eqnarray}
m &=&\lim_{N\to\infty}\meanv{\frac{\bm{1}\cdot\bxi}{N}}=  \lim_{N\to\infty} \meanv{Q(\bxip,\bxi)}^{\hat{}}\\
q &=& \lim_{N\to\infty} \meanv{Q(\bxi^1,\bxi^2)}=\lim_{N\to\infty} \meanv{Q(\bxi^1,\bxi^2)}^{\hat{}},
\end{eqnarray}
where $\bxi^1,\bxi^2$ are two replicas of the systems, i.e. two independent configurations sampled from the posterior distribution (\ref{eq:hopfield}) with the same data $\bS$. 
\end{ansatz}

The details about the derivation of Conjectures \ref{conj:RS} and \ref{conj:RSparameters} are provided in Appendix \ref{appendix:replica}. As in Theorem \ref{thm:1} the free energy is given as the solution of a variational principle in terms of the model' s  order parameters: in that case the overlap with the examples $p^\mu\sim Q(\bm{s}^\mu,\bxi)$. In this case, because of the high temperature ($\beta<1$) the system never aligns with any specific examples ($\bm{p}=0$) and this overlap does not play a role in the free energy principle. Nevertheless, 
the signal (toward the teacher pattern) carried by an extensive number of examples can become macroscopic and could bring to non zero system's magnetization and system's overlap ($m$ and $q$) that in fact emerge as the two natural order parameters. Eqs. (\ref{eq:m}) and (\ref{eq:q}) are similar to those of the standard Hopfield model \cite{coolen2005theory}, with a random gaussian field of a different variance, still proportional to the load $\gamma$. Moreover the signal term $\hat{m}$ doesn't point towards the examples but towards the teacher pattern and it is proportional to $\gamma$, thus showing the beneficial effect of the training set size.  
 
It is important to recall that we are studying the problem in which the student machine is exactly as the teacher one (same architecture) and also the temperatures are the same. Therefore by construction the model satisfies the Nishimori conditions, in particular
\begin{equation}
\meanv{Q(\bxi^1,\bxi^2)}^{\hat{}}=\meanv{Q(\bxip,\bxi)}^{\hat{}}.
\end{equation}
Infact we have at $\lambda=0$ and $\bm{h}=0$ that 
\begin{equation}
\begin{split}
    \meanv{Q(\bxip,\bxi)}^{\hat{}}
    & = 
   2^{-N}\sum_{\bxip,\bS,\bxi} Q(\bm{\hat{\xi}},{\bm{\xi}}) \hat{P}(\bm{\xi}|\bm{\mathcal{S}}) P(\bS|\bxip) \\
    &= 
    \sum_{\bxip,\bS,\bxi}  Q(\bm{\hat{\xi}},{\bm{\xi}}) \hat{P}(\bm{\xi}|\bm{\mathcal{S}}) P(\bm{\hat{\xi}}|\bm{\mathcal{S}}) P(\bm{\mathcal{S}})\\
    &=2^{-N}\sum_{\bS,\bm{\xi}^1,\bm{\xi}^2} Q(\bm{\xi}^1, \bm{\xi}^2) P(\bm{\xi}^1|\bm{\mathcal{S}}) P(\bm{\xi}^2|\bm{\mathcal{S}}) P(\bm{\mathcal{S}}|\bxip) 
    =:\meanv{Q(\bxi^1,\bxi^2)}^{\hat{}}.
\end{split}    
\end{equation}
For this reason, according to Conjecture 2,  we expect that the solution of the saddle point equations satisfies $m=q$ and at the same time, see Eqs. (\ref{eq:mh},\ref{eq:qh}) $\hat{m}=\hat{q}$. It is easy to show that this is in fact a solution of Eqs. (\ref{eq:m}-\ref{eq:qh})
by using the identity
\begin{align}
   \int D\mu(z) \tanh(\hat{m}+z\sqrt{\hat{m}}) = \int D\mu(z) \tanh^2(\hat{m}+z\sqrt{\hat{m}}).
\end{align}
We checked numerically this is a stable solution. This condition indicates the absence of a spin-glass region ($m=0$, $q>0$). Analogously it is easy to show \cite{nishimori2001statistical} that overlap and magnetization have the same distribution. The expected self-averaging of the system's magnetization motivates the belief that the model is replica symmetric and that conjectures \ref{conj:RS} and \ref{conj:RSparameters} therefore hold \cite{alberici2021multi,alberici2021solution}. 
\begin{figure}[ht]
\centering
\includegraphics[width=.5\textwidth]{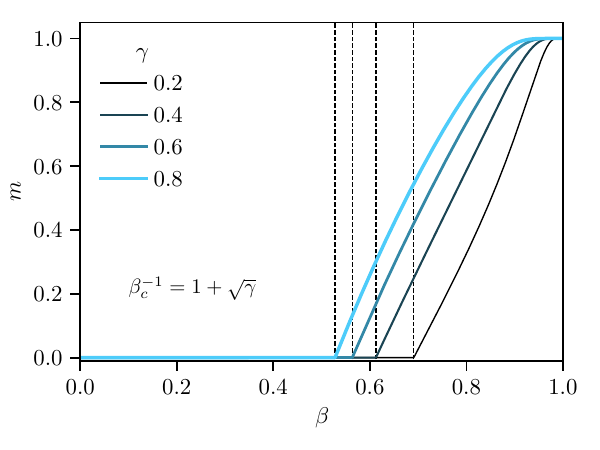}\hfil
\includegraphics[width=.5\textwidth]{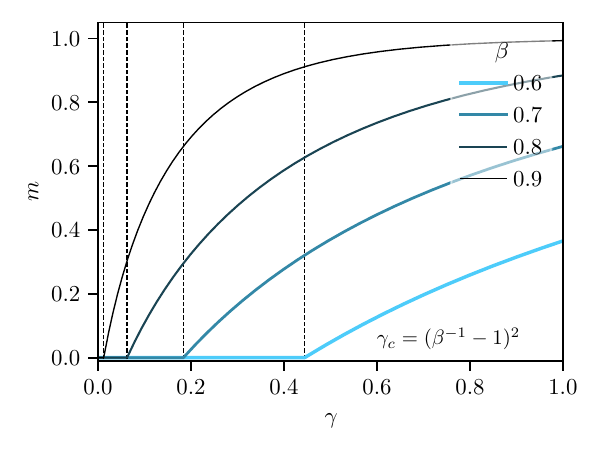}
\caption{Learning performance with a noisy ($\beta<1$) but extensive  ($M=\gamma N$) dataset on the \textit{Nishimori} line. \textit{Left}: the system's magnetisation $m$ is shown as a function of the inverse temperature $\beta$  and for different dataset size $\gamma$. \textit{Right}: the magnetisation $m$ is shown as a function of $\gamma$ and different inverse temperatures $\beta$. The inferred pattern's quality displays a second order phase transition. Moreover it increases with $\gamma$ and decreases with the dataset noise $\beta^{-1}$.}\label{fig:alta_temperatura}
\end{figure}
Figs. \ref{fig:alta_temperatura} show the value of the magnetization, i.e. the learning performance, as a function of $\beta$ and $\gamma$. It is evident the occurrence of a  second order phase transition from a paramagnetic region where the only solution is $m=q=0$ to a ferromagnetic region where $m=q>0$ and learning is feasible. The phase transition occurs at a critical temperature $\beta_c(\gamma)$ if we fix the size of the dataset $\gamma$ or equivalently at a critical size $\gamma_c(\beta)$ for a given level of the temperature. The critical line can be obtained analytically by studying the reduced equation 
\begin{equation}
    m=\int D\mu(z) \tanh(\hat{m}(m)+z\sqrt{\hat{m}(m)}) 
\end{equation}
where 
\begin{equation}
\hat{m}(m)= \frac{\gamma \beta^2 m}{(1-\beta)(1-\beta+\beta m)}.
\end{equation}
By expanding for small values of the magnetization we get
\begin{equation}
m= \frac{\gamma \beta^2}{(1-\beta)^2} m + o(m)
\end{equation}
from which the bifurcation must be at $\gamma \beta^2/(1-\beta)^2 =1$, i.e. at
\begin{equation}
    \beta_c^{-1}(\gamma):=1+\sqrt{\gamma}\ \ \ \ \ \gamma_c(\beta):=\frac{1-\beta^2}{\beta^2}.
\end{equation}

As expected the critical size is an increasing function of the data temperature $\beta^{-1}$, thus of the data correlation with the signal. What is interesting is that, despite we are in the regime in which each example $s^\mu$ does not share macroscopic correlation with the signal ($m_0(\beta)=0$), still the machine is able to retrieve it $m>0$ as soon as the dataset is sufficiently large. It means that the dataset contains enough information but divided in many ( an extensive number of examples) small (poorly correlated examples) pieces.

\begin{figure}[ht]
\centering
\includegraphics[width=9cm]{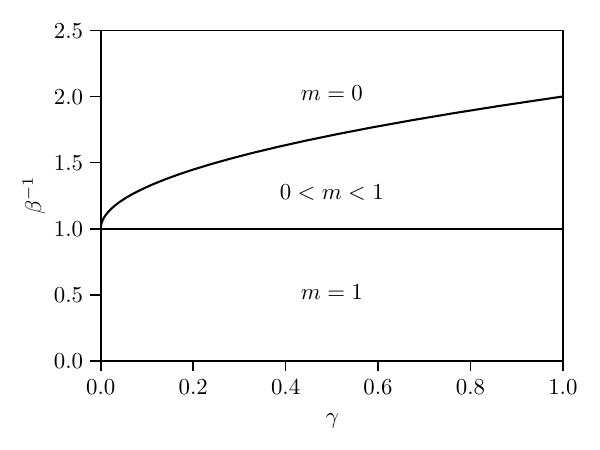}
\caption{Phase diagram of the model on  the \textit{Nishimori} line. For $\beta^{-1} > 1+\sqrt{\gamma}$, the student machine is in the paramagnetic phase with $ m = 0$, where learning is impossible. Conversely it enters a learning phase where it can infer the original pattern by generalization from a  sea of corrupted examples that the teacher provides. For $\beta^{-1}<1$ each example is highly informative and the learning performance is optimal ($m=1$).}
\label{fig:diagram_dual}
\end{figure}
\noindent In Fig. \ref{fig:diagram_dual} the phase diagram is shown. At low temperature we know from the previous section that learning is always possible, in particular a perfect retrieval of the teacher's pattern ($m=q=1$) is achieved when $M\to \infty$. At high temperature, i.e. poorly informative dataset, a paramagnetic region where learning is not possible is separated from a ferromagnetic region  where the signal inference is still possible leveraging on a sufficiently large dataset.

It is interesting to note that the critical line coincides with the paramagnetic to spin-glass transition line in the standard Hopfield model: the two systems becomes frozen in  the same moment, when the signal from the patterns become macroscopic and prevails w.r.t. the temperature noise.  In the standard Hopfield model different patterns led to different signals because they were independent and unbiased. For this reason the system got confused by increasing their number (the network load) and enters a spinglass regime.  In our model each example carries a vanishing but non zero bias toward the signal, thus this bias becomes macroscopic by increasing extensively the size of the dataset and the system enters a ferromagnetic, ordered phase, where learning is possible.

\subsection{Inference temperature vs dataset noise}

The assumption that the student's machine is exactly equal to the teacher one is not realistic. The interesting research question is in fact related to the representation performance of a particular learning machine in relation to different possible data structures. To this aim, in this section we investigate one possible miss-matching between data (i.e. teacher machine) and student machine: the one related to the use of an inference temperature which differs from the real  generating temperature of the data. We therefore assume that the training set is generated at an inverse temperature $\hat{\beta}$, i.e.
\begin{equation}
    P(\bS|\bxip)= \prod_{\mu=1}^M z^{-1}\exp\left(\frac{\hat{\beta}}{N}\sum_{i<j}^N\hat{\xi_i}\hat{\xi_j}s^\mu_is^\mu_j\right),
\end{equation}
while the student patterns are still sampled at an inverse temperature $\beta$ as in Eq. (\ref{eq:hopfield}), which represents the posterior distribution of the learning problem with a miss-matched prior. It represents the more realistic situation in which the dataset noise $\hat{\beta}^{-1}$ is unknown and the machine is trained with a different weights regularization $\beta$.

\noindent As long as $M$ is finite, following the proof of Theorem \ref{thm:1} at $\lambda=0$, it holds that in the ferromagnetic gauge
\begin{equation}
    -\beta f =\lim_{N\to\infty} \frac 1 N \mathbb{E}_{\hat{\beta},\bm{0}^+} \log Z(\bS)=\sup_{\bm{p}\in\mathbb{R}^M} \hat{g}(\bm{p}).
\end{equation}
The free energy density trial function
\begin{equation}\label{eq:freemism}
\hat{g}(\bm{p}):= \log 2 -\frac{\beta  \bm{p}^2}{2} +\meanv{\log\cosh \left( \beta \bs\cdot\bm{p} \right)}_{\bs,\hat{\beta}},
\end{equation}
where now the random vector $\bs\in \{-1,1\}^M$ has i.i.d. random entries with mean $m_0(\hat{\beta})$, depending on the generating temperature.
By extremizing Eq. (\ref{eq:freemism}), $\bm{p}$ has to be a solution of
\begin{equation}\label{eq:self_mism}
    \bm{p}=\meanv{\bs \tanh(\beta \bs\cdot\bm p)}_{\bs,\hat{\beta}}=
\frac{\sum_{\bm{s}} \bs\  e^{\hat{\beta} m_0\bm{1}\cdot \bm{s}} \tanh(\beta \bm{s}\cdot\bm{p})}{\left(2 \cosh(\hat{\beta} m_0)\right)^M}, 
\end{equation}
from which, the learning performance can be derived as
\begin{eqnarray}\label{eq:m_mism}
\lim_{N\to\infty}\meanv{Q(\bxip,\bxi)}^{\hat{}}=m= \meanv{\tanh(\beta \bs\cdot\bm p)}_{\bs,\hat{\beta}}=
\frac{\sum_{\bm{s}} e^{\hat{\beta} m_0\bm{1}\cdot \bm{s}} \tanh(\beta \bm{s}\cdot\bm{p})}{\left(2 \cosh(\hat{\beta} m_0)\right)^M}.
\end{eqnarray}
Similarly to Eq. (\ref{eq:paraNishi}) it holds the following
\begin{proposition}\label{prop:5}
As long as 
\begin{equation}
\beta^{-1}>1+(M-1)m_0^2(\hat{\beta})
\end{equation}
the only solution of Eqs. (\ref{eq:self_mism}) is $\bm{p}=0$. As a consequence from Eq. (\ref{eq:m_mism}) $\bm{m}=0$.
\end{proposition}
\begin{proof}
It is sufficient to see that
\begin{eqnarray}
\bm{p}^2&=&\meanv{\bm{p}\cdot\bs \tanh(\beta \bs\cdot\bm p)}_{\bs,\hat{\beta}}\leq \beta \meanv{(\bm{p}\cdot\bs)^2}_{\bs,\hat{\beta}} =\beta \sum_{\mu,\nu}p^\mu p^\nu \meanv{s_\mu s_\nu}_{\bs,\hat{\beta}}\nonumber\\
&=&\beta \sum_{\mu,\nu}p^\mu p^\nu (\delta_{\mu\nu}+(1-\delta_{\mu\nu}) m^2_0(\hat{\beta}))
=\beta(1-m_0^2(\hat{\beta}))\bm{p}^2+\beta m_0^2(\hat{\beta})(\sum_\mu p^\mu)^2\nonumber\\
&\leq&\beta(1+(M-1)m_0^2(\hat{\beta}))\bm{p}^2.
\end{eqnarray}
\end{proof}
As soon as the inference temperature drops below the threshold provided by Proposition \ref{prop:5} other solutions of Eqs. (\ref{eq:self_mism}) appear according to the following 
\begin{proposition}\label{prop:mism_inst}
As long as $\beta^{-1}<1+(M-1)m_0^2(\hat{\beta})$ the global maximum of $\hat{g}(\bm{p})$ is attained far from $0$. Moreover there exist solutions of Eqs. (\ref{eq:self_mism}) of the form $p^\mu=\pm \bar{p}$, $\forall \mu=1,\ldots, M$ where $\bar{p}>0$ is unique.
\end{proposition}
\begin{proof}
It is sufficient to study the reduced equation
\begin{equation}
p=\meanv{s_1\tanh\left(\beta p\sum_\mu s_\mu\right)}_{\bs,\hat{\beta}}:=f(p;\beta,\hat{\beta}).
\end{equation}
The function $f$ is odd and bounded in $(-1,1)$. Moreover its derivative in $p=0$ is 
\begin{equation}
\frac{\partial f}{\partial p}|_{p=0}= \beta \meanv{s_1\sum_\mu s_\mu}_{\bs,\hat{\beta}}=\beta(1+(M-1)m_0^2(\hat{\beta})).
\end{equation}
As soon as this derivative becomes larger than $1$ the function must intersect the bisector far from the origin in at least two symmetric points $\pm \bar{p}$, $\bar{p}>0$. At the same time, the unique maximum of $\hat{g}(\bm{p})$ restricted to $p^\mu=p$, $\forall \mu=1,\ldots,M$, is attained in $\bar{p}>0$, see the proof of Proposition \ref{prop:3}. This proves the uniqueness of $\bar{p}$ and the instability of $\bm{p}=0$.
\end{proof}
Note that Proposition \ref{prop:mism_inst} does not prove that the maximum of $\hat{g}(\bm{p})$ is in $\bm{p}=\pm\bar{p}\bm{1}$ because there could exist other solutions of Eqs. (\ref{eq:self_mism}) that are not homogeneous. For example as long as $\hat{\beta}<1$ (thus $m_0(\hat{\beta})=0$) and $\beta>1$, there exist solutions in which the system is aligned with a single example (pure states), i.e. $p^\mu=\bar{p}_1\neq 0$, $p^\nu=0$ $\forall \nu\neq\mu$.
In fact it is sufficient to fix $\bar{p}_1$ as the solution of 
\begin{equation}
p=\meanv{s_1\tanh(\beta p s_1)}_{\bs,0},
\end{equation}
i.e. $\bar{p}_1=\pm m_0(\beta)$. Analogously there could exist solutions in which the system is homogeneously aligned with a subset $E_k\subset\{1,\ldots,M\}$, $|E_k|=k$, of the examples (mixed states), i.e. $p^\mu=\bar{p}_k\neq 0$ $\forall \mu\in E_k $, $p^\nu=0$ $\forall \notin E_k$: in this case $\bar{p}_k$ has to be the solution of
\begin{equation}
p=\meanv{s_1\tanh(\beta p \sum_{\mu=1}^k s_k )}_{\bs,0}.
\end{equation} 
However, if the value of the inference temperature is not too low with respect to the dataset noise, then Proposition \ref{prop:2} can be generalized according to the following
\begin{proposition}\label{prop:hom}
As long as 
\begin{equation}
\beta^{-1}> 1-m^2_0(\hat{\beta})
\end{equation}
the solutions of Eqs. (\ref{eq:self_mism}) have equal components.
\end{proposition}
\begin{proof}
Following the proof of Proposition \ref{prop:2} it can be proved that for $\mu\neq\nu$ it holds
\begin{equation}
|p^\mu-p^\nu|\leq \beta (1-m_0^2(\hat{\beta}))|p^\mu-p^\nu|.
\end{equation}
Therefore as long as $\beta (1-m_0^2(\hat{\beta}))<1$ the only solution is homogeneous.
\end{proof}
In the region  between the instability condition of Proposition \ref{prop:mism_inst} and the homogeneity condition of Proposition \ref{prop:hom}, i.e.
\begin{equation}
1-m_0^2(\hat{\beta})<\beta^{-1}< 1+(M-1)m_0^2(\hat{\beta}),
\end{equation}
which is non empty only if $\hat{\beta}>1$,
the global maximum of $\hat{g}(\bm{p})$ is attained in the two symmetric point $\bm{p}=\pm \bar{p}\neq 0$ and consequently the system is magnetized, i.e. it is aligned with the signal since
\begin{equation}
m=\pm\meanv{\tanh(\beta \bar{p}\sum_\mu s_\mu )}_{\bs,\hat{\beta}}\neq 0.
\end{equation}
Conversely if $\hat{\beta}<1$ the value of the system's magnetization given by Eq. (\ref{eq:m_mism}), i.e. the learning performance, is always zero because $m_0(\hat{\beta})=0$ independently from the value of $\bm{p}$. The phase diagram of the model, in terms of the value of $m$ and $\bm{p}$ is shown in Fig. \ref{fig:mismatch1}, where four different regions appear:
\begin{itemize}
    \item Paramagnetic (P) region: $\bm{p}=\bm{0}$ and $m=0$;
    \item Signal retrieval (sR) region: $\bm{p}= \bar{p}\bm{1}$, $m>0$ ;
    \item Example retrieval (eR) region: $p\neq0$, $m=0$;
    \item Mixed retrieval (mR) region: $\bm{p}\neq \bm{0}$, $m>0$.
\end{itemize}

\begin{figure}[ht]
\centering
\includegraphics[width=.5\textwidth]{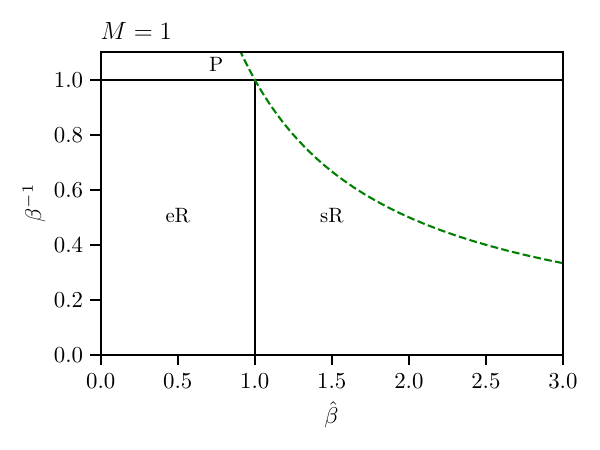}\hfil
\includegraphics[width=.5\textwidth]{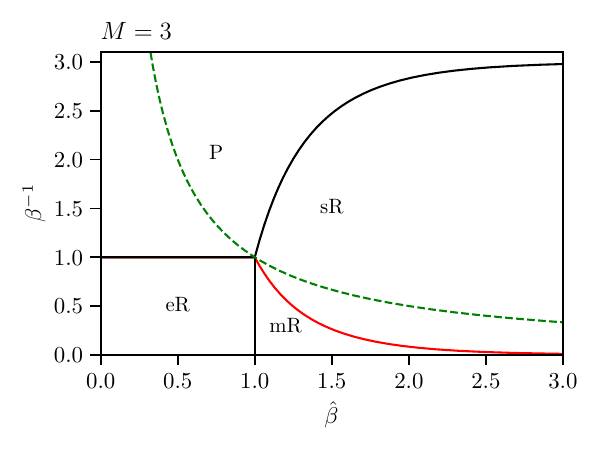}
    \caption{Phase diagram of the model in the case of \textit{mismatched} setting and finite $M$ ($M=1$ on the left, $M=3$ on the right) in terms of the dataset information $\hat{\beta}$ and the inference temperature $\beta^{-1}$. According to the values of $m$ and $\bm{p}$ solutions of Eqs. (\ref{eq:m_mism},\ref{eq:self_mism}) four different regimes appear: in the paramagnetic (P) regime $m=0$, $\bm{p}=0$; in the example retrieval (eR) regime $\bm{p}\neq \bm{0}$ but $m=0$; in the signal retrieval (sR) regime $\bm{p}=\bar{p}\bm{1}$ is homogeneous and $m>0$; in the mixed retrieval (mR) regime it is $\bm{p}\neq\bm{0}$ and $m>0$. Only in the sR and mR regimes the machine can learn the original signal and the learning performance monotonically increases with $\hat{\beta}$.
    The Nishimori line $\hat{\beta}=\beta$ is shown in green.  
    }\label{fig:mismatch1}
\end{figure}
In the paramagnetic region the inference temperature is too high and the machine neither stores the examples $\bm{p}=\bm{0}$ nor can learn the signal $m=0$. In the eR region the inference temperature is low enough to allow the storage of the examples $\bm{p}\neq\bm{0}$ but they are not enough informative to allow signal learning. The stability of the different $\bm{p}\neq 0$ solutions can be investigated exactly as in the case of the standard Hopfield model (see \cite{coolen2005theory}) but every one of them leads to a poor learning performance $m=0$. Conversely in the sR region the dataset noise is low and the stored examples informative: in this region the machine can learn the signal by example memorization. Moreover, since $\bm{p}=\bar{p}\bm{1}$ in this region, it seems that the machine is working by uniformly using all the examples in a kind of early attempt of learning by generalization. Interestingly the machine can work efficiently even at temperatures $\beta^{-1}$ higher then the dataset noise $\hat{\beta}^{-1}$. However the learning performance $m$ increases monotonically by lowering the inference temperature. Finally in the mR region, different stable solutions for $\bm{p}\neq 0$ coexists, everyone leading to a positive global magnetization $m$, where typically the globally stable one is non-homogeneous, with $\max |p_i|>\bar{p}$. It seems to suggest that in this regime the machine prefers to learn the signal mainly  leveraging on the high amount of information carried by a single (or few) examples, typical behaviour of a learning by memorization.

As in the previous section we expect that increasing the size of the dataset proportionally to the size of the system, i.e. $M=\gamma N$, it could be possible to retrieve the original pattern even when the examples are particularly noisy, i.e. their generating temperature is higher than $1$. In this regime it is possible to generalize the Conjecture \ref{conj:RS} and obtain the replica symmetric approximation of the limiting free energy in terms of the two temperatures $\beta$ and $\hat{\beta}$ as 
\begin{equation}
    -\beta f^{RS}=\operatorname{Extr}_{m,\hat{m},p,\hat{p},q,\hat{q}} \hat{f}(p,q,m,\hat{p},\hat{q},\hat{m}),
\end{equation}
where
\begin{align}
    \hat{f}(p,q,m,\hat{p},\hat{q},\hat{m})=
     & \ln 2-\frac{\gamma}{2} \ln\left((1-\hat{\beta})(1- \beta +\beta q) \right) +\frac{\gamma}{2}\frac{\beta(1-\hat{\beta})q+\hat{\beta}\beta m^2}{(1-\hat{\beta})(1- \beta +\beta q)}  \nonumber\\
     &+\frac{\hat{q}q}{2} -\hat{m}m - \frac{\hat{q}}{2}-\hat{p}p+ \frac{\beta}{2}p^2+ \meanv{ \int \mathcal{D}z \ln \cosh \left(\hat{p} s + z\sqrt{\hat{q}} + \hat{m}\right)}_s.
\end{align}
The RS saddle point equations read as 
\begin{align}
    m&=\int D\mu(z) \meanv{\tanh(\beta p s + \hat{m}+z\sqrt{\hat{q}})}_s \label{eq:mism_sp1}\\
    q &= \int D\mu(z) \meanv{\tanh^2(\beta p s +\hat{m}+z\sqrt{\hat{q}})}_s \label{eq:mism_sp2}\\
    p &= \int D\mu(z) \meanv{s\tanh(\beta p s +\hat{m}+z\sqrt{\hat{q}})}_s \label{eq:mism_sp3}
\end{align}
where $s$ is an auxiliary Rademacher random variable with symmetric distribution and
\begin{align}
    \hat{m}&= \frac{\gamma \hat{\beta}\beta m}{(1-\hat{\beta})(1-\beta+\beta q)}\nonumber\\
    \hat{q} &=  \frac{\gamma \hat{\beta}\beta^2 m^2 +\gamma \beta^2 q(1-\hat{\beta}) }{(1-\hat{\beta})(1-\beta+\beta q)^2}.
\end{align}
The order parameter $p$ has to be interpreted as the overlap between the student pattern and the examples, i.e.
\begin{equation}
    p =\lim_{N\to\infty} \meanv{Q(\bs^\mu,\bxi)}= \lim_{N\to\infty} \meanv{Q(\bs^\mu,\bxi)}^{\hat{}}.
\end{equation}
Equations (\ref{eq:mism_sp1},\ref{eq:mism_sp2},\ref{eq:mism_sp3}) reduce to those of Conjecture \ref{conj:RS} when $\hat{\beta}=\beta$ and $p=0$. In fact in that case, since $\hat{\beta}=\beta<1$, the system is never aligned with any example. Conversely if $\beta\neq\hat{\beta}$, even if $\hat{\beta}<1$, the inference temperature $\beta^{-1}$ could be in principle low enough to allow example retrieval. It is important to stress however that, since the examples are only weakly correlated with the signal,  this situation would prevent the system to be aligned with the original pattern. 
\begin{figure}
    \centering
    \includegraphics[scale=0.68]{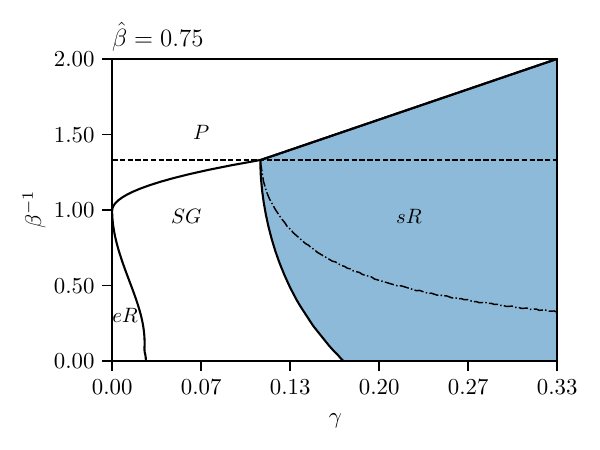}
    \includegraphics[scale=0.68]{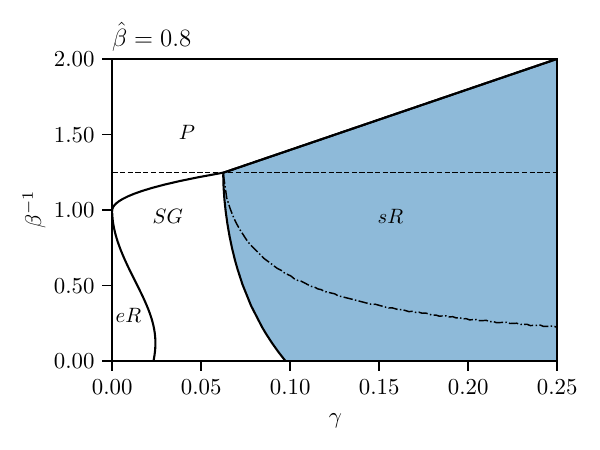}
    \includegraphics[scale=0.68]{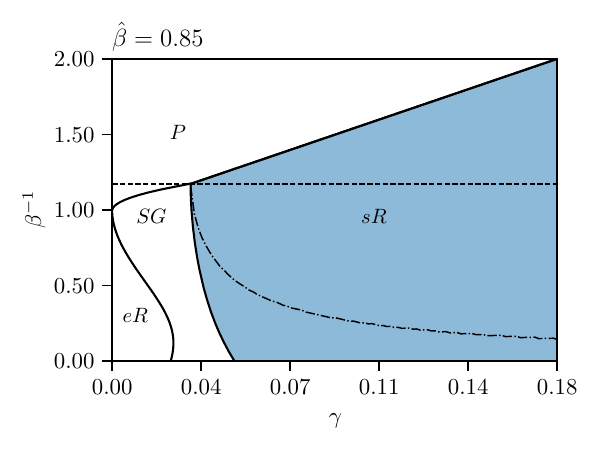}
    \includegraphics[scale=0.68]{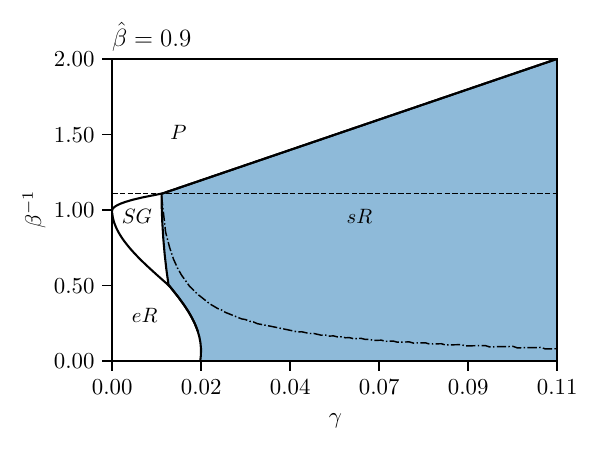}\caption{Phase diagram of the model in the mismatched setting, where $\hat{\beta} \neq \beta$, and extensive dataset $M=\gamma N$. In the paramagnetic (P) and  spin glass (SG) regions learning is impossible. For higher values of the dataset size, the machine enters a signal retrieval region (sR) where it learns by generalizations. In this region the learning performance $m$ has a maximum (dot-dash line) for a specific value of the inference temperature. In particular if $\beta^{-1}$ gets too low the machine enters the example retrieval (eR) region where it is forced to work by memorization when this approach is inefficient for learning. The dotted line is the Nishimori condition $\beta=\hat{\beta}$.
    }
    \label{fig:phasediag_mismatched_extensive}
\end{figure}
In Figure \ref{fig:phasediag_mismatched_extensive} the phase diagram of the model is shown for different values of the generating temperature $\hat{\beta}$. Four different regions appear, depending on the properties of the globally stable solution of Eqs. (\ref{eq:mism_sp1},\ref{eq:mism_sp2},\ref{eq:mism_sp3}):
\begin{itemize}
    \item Paramagnetic (P) region: $m=q=p=0$;
    \item Signal retrieval (sR) region: $m\neq 0$, $q>0$, $p=0$ ;
    \item Example retrieval (eR) region: $p\neq0$, $q>0$, $m=0$;
    \item Spin Glass (SG) region: $m=p=0$, $q>0$.
\end{itemize}
Only in the $sR$ phase student and teacher patterns are correlated and the learning performance is positive. In all other phases the student patterns is uncorrelated with the signal, being a random guess (P region), aligned with a noisy example (eR) or aligned with a spurious low energy state (SG region). 
From the high temperature (P) phase to the low temperature (SG or sR) phases a second order phase transition occurs when the system's overlap $q$ detaches from zero. The transition line can be obtained by expanding eq. ($\ref{eq:mism_sp1}$) or  eq. (\ref{eq:mism_sp2}) depending on the magnetization $m$ behavior. For small values of both $m$ and $q$ (P to sR), from eq. (\ref{eq:mism_sp1}) we get
\begin{equation}
m= \frac{\gamma \hat{\beta \beta}}{(1-\beta)(1-\hat{\beta})}m + O(mq,m^2)
\end{equation}
that gives the instability condition
\begin{equation}
\frac{\gamma \hat{\beta \beta}}{(1-\beta)(1-\hat{\beta})}=1 \ \ \implies \beta^{-1}=1+\gamma\frac{\hat{\beta}}{1-\hat{\beta}}.
\end{equation}
On the other hand if $m=0$ across the transition (P to SG), by expanding Eq. (\ref{eq:mism_sp2}) for small values of $q$ we get
\begin{equation}
q = \frac{\gamma\beta^2}{(1-\beta)^2} q +O(q^2)
\end{equation}
that gives the usual  instability condition
\begin{equation}
\frac{\gamma\beta^2}{(1-\beta)^2}=1\ \ \ \implies \beta^{-1}=1+\sqrt{\gamma}
\end{equation}
Therefore, starting from the paramagnetic region and decreasing the inference temperature a phase transition occurs as soon as one of the two instability condition is satisfied, i.e. at
\begin{equation}
\beta^{-1}(\gamma;\hat{\beta}):=\max\left\{1+\sqrt{\gamma}; 1+\gamma\frac{\hat{\beta}}{1-\hat{\beta}}   \right\}.
\end{equation}
Interestingly the two lines cross exactly at the point $(\beta,\gamma)=(\hat{\beta}, (1-\hat{\beta})^2/\hat{\beta}^2) )$, in agreement with the usual property of the Nishimori line of crossing a triple critical point. For smaller values of $\gamma$ the transition is towards a SG regime, while for higher values of $\gamma$ the transition is towards a sR region, where learning is easy. The other two transitions, from sR to SG (second order) and from SG to eR (first order) can be found  numerically and shown in Figure \ref{fig:phasediag_mismatched_extensive}.
The phase diagram  shows a non monotone behavior of the learning performance in terms of the inference temperature: if $\beta^{-1}$ is too high the learning performance is that of a random guess (P phase), if $\beta^{-1}$ is too low the learning performance can deteriorate because of the emergence of low energy configurations that are uncorrelated with the signal (they can be either correlated with the examples, eR region, or completely uncorrelated with both signal and examples, SG region). In particular the eR phase identifies a regime in which the machine is forced (through $\beta$) to work by memorization ($p>0$) in a situation where this approach is highly inefficient for learning. By increasing $\gamma$ the memory storage limit of the machine becomes beneficial for the occurrence of a region where learning is possible by generalization. Interestingly the phase diagram of Fig. \ref{fig:phasediag_mismatched_extensive} seems qualitatively similar to that of Fig. \ref{fig:mismatch1} where in both cases the x-axis measures  the amount of information contained in the dataset.

Finally note that Eqs. (\ref{eq:mism_sp1})-(\ref{eq:mism_sp3}) becomes exactly those of the classicl Hopfield model if we force $m=0$. This means that a purely SG solution always exists below the inference temperature $1+\sqrt{\gamma}$, which is only locally stable inside the sR region. In this case a Monte-Carlo simulation,  performed at a low temperature, can remain trapped in the locally stable spin-glass state. Fortunately, this occurrence can be avoided by using Simulated Annealing and lowering the temperature very slowly; this is possible because the critical temperature for signal retrieval is higher than $1+\sqrt{\gamma}$. A similar strategy that leverages on the hierarchy between temperatures of ergodicity breaking is investigated in  \cite{angelini2023limits}.


\section{Proofs}

In this section the proofs of the main results  are provided, together with some technical results, in the form of lemmas and propositions, needed for the proofs. 
Since it will be used many times in the rest of the section we recall a standard result about the Fourier decomposition of a function of boolean ($\pm1$) variables. Given $\Lambda=\{1,\ldots,M\}$, any function $f:\{-1,1\}^M\to\mathbb{R}$ can be decomposed as \begin{equation}
f(\bm{s}) = \sum_{X\subset \mathcal{P}(\Lambda)} \left\langle f, s_X\right \rangle s_X,
\end{equation}
where $s_X=\prod_{\mu\in X} s_\mu$ and $\left\langle f,g \right \rangle= 2^{-M} \sum_{\bm s}f(\bm s) g(\bm s)=\left\langle f g\right\rangle_{\bm{s},0}$.

\begin{proof}{of Theorem \ref{thm:1}}

By using gaussian linearization we can write the partition function as
\begin{equation}
\begin{split}
    Z(\bm{\mathcal{S}}) & := \sum_{\bm{\xi}} \exp\left(\frac{\beta}{2N} \sum_{\mu=1}^M \sum_{i,j=1}^N s_i^\mu s_j^\mu \xi_i\xi_j+ \lambda\sum_{i=1}^N\xi_i\right)\\
    & = \sum_{\bm{\xi}}\int D\mu(\bm{z}) \exp\left(\sqrt{\frac \beta N}\sum_{\mu=1}^M\sum_{i=1}^{N}s^\mu_i\xi_iz^\mu+ \lambda\sum_{i=1}^N\xi_i \right).
\end{split}
\end{equation}

where  $\mathbb{R}^M\ni \bm{z}\sim \mathcal{N}(0,\mathbb{I})$. By making a change of variables $p^\mu=z^\mu/\sqrt{N\beta}$ we have
\begin{align}
    Z(\bm{\mathcal{S}}) & \propto \sum_{\bm{\xi}} \int  d\bm{p}\exp\left(-\beta N \bm{p}^2/2+\sum_{i=1}^N (\beta \bm{s_i}\cdot\bm{p}+\lambda )\xi_i \right) \label{eq:magMat_sella}\\
   & = \int  d\bm{p} \ \exp \left( N\left( \log 2-\frac{\beta}{2} \bm{p}^2 + \frac{1}{N}\sum_{i=1}^N\log\cosh(\beta \bm{s}_i\cdot\bm{p} +\lambda)\right)\right) \notag\\
   &= \int  d\bm{p} \ e^{ N g_N(\bm{p},\bS)}\nonumber
\end{align}
We recall that the random vector of the examples $\bS$ is drawn from Eq. ($\ref{eq:pattern_d}$): using Lemma \ref{lemma:1} on the function $F:\{-1,1\}^M\to\mathbb{R}$, $ F(\bm{s})=\log\cosh(\beta \bm{s}\cdot\bm{p} +\lambda)$, it holds 
\begin{equation}\label{eq:limsella}
\lim_{N\to\infty}\mathbb{E}_{\beta,h\bm{\epsilon}}\ g_N(\bm{p},\bS)=g(\bm{p}):=
    \log2 -\frac{\beta\bm{p}^2}{2} +\left\langle \log\cosh(\beta  \bm{s}\cdot \bm{p} + \lambda ) \right\rangle_{\bs,\bm{\epsilon}},
\end{equation}
where $\meanv{.}_{\bs,\bm{\epsilon}}$ denotes the  mean-field expectation w.r.t. the random vector $\bs\in\{-1,1\}^M$, whose entries are independent with mean $\meanv{s^\mu}_{\bs,\bm{\epsilon}}=m_0(\beta,\epsilon^\mu h)$ and   $m_0(\beta, h)$ is the unique solution of $m_0=\tanh(\beta m_0 +h)$. For any compact set $K\subset\mathbb{R}$, using Lemma  \ref{lemma:2} on $F:K\times\{-1,1\}^M\to\mathbb{R}$, $ F(\bm{p},\bm{s})=\log\cosh(\beta \bm{s}\cdot\bm{p} +\lambda)$, it holds that
\begin{equation}\label{eq:grandinumeriM}
   g_N(\bm{p},\bS) \overset{p }{\longrightarrow}  g(\bm{p}), 
\end{equation}
uniformly in $K$, meaning that $g_N(\bm{p},\bS)$ is self-averaging uniformly in any compact.
Thanks to Lemma \ref{lem:bound-delta.gn}, $g_N$ and $g$ satisfy the conditions necessary for a generalized saddle point approximation, i.e. Proposition \ref{pr:sp1} and Proposition \ref{pr:sp2}, so that
\begin{equation}
    -\beta f =
    \lim_{N\to\infty}\frac 1 N \log \mathbb{E}_{\beta,h\bm{\epsilon}}\int_{\mathbb{R}^M}  d\bm{p} \ e^{ N g_N(\bm{p};\bS)}
    =\sup_{\mathbb{R}^M} g(\bm{p}).
\end{equation}
It is easy to show that $-\beta f$ does not depend on $\epsilon$ by using the transformations $s^\mu\to\epsilon^\mu s^\mu$ and $p^\mu\to\epsilon^\mu p^\mu$. This concludes the proof if one chooses $\epsilon = \bm{1}$.

\end{proof}

\begin{lemma}\label{lemma:1}
Let $\bs_i\in\{-1,1\}^M$ the $i$-th marginal of $\bS=\{\bs_i\}_{i=1}^N\in\{-1,1\}^{NM}$, distributed according to (\ref{eq:pattern_d}). For any function $F: \{-1,1\}^M\to\mathbb{R}$ it holds 
\begin{equation}
\lim_{N\to\infty}\mathbb{E}_{\beta,h\bm{\epsilon}}\ F(\bs_i)=\left\langle F(\bs) \right\rangle_{\bs,\bm{\epsilon}},\ \ \ \forall i=1,\ldots,N,
\end{equation}
where $\meanv{.}_{\bs,\bm{\epsilon}}$ denotes the expectation w.r.t. the random vector $\bs\in\{-1,1\}^M$ whose entries are independent with mean $\meanv{s^\mu}_{\bs,\bm{\epsilon}}=m_0(\beta,\epsilon^\mu h)$ and   $m_0(\beta, h)$ is the unique solution of $m_0=\tanh(\beta m_0 +h)$. 
\end{lemma}

\begin{proof}
By Fourier decomposition $F$ can be written as 
\begin{equation}
F(\bs)=\sum_{X\subset \mathcal{P}(\Lambda)} c_X \prod_{\mu\in X}s^\mu.
\end{equation}
As a consequence, for any factorized probability $\pi(\bs)=\prod_{\mu=1}^M\pi_\mu(s^\mu)$ it holds   
\begin{equation}
\meanv{F(\bs)}_\pi=\sum_{X\subset \mathcal{P}(\Lambda)} c_X \meanv{\prod_{\mu\in X}s^\mu}_{\pi}=\sum_{X\subset \mathcal{P}(\Lambda)} c_X \prod_{\mu\in X}\meanv{s^\mu}_{\pi_\mu}.
\end{equation}
Therefore, since the Fourier decomposition has a finite number of terms, it is 
\begin{eqnarray}
\lim_{N\to\infty}\mathbb{E}_{\beta,h\bm{\epsilon}}\ F(\bs_i)&=&
\sum_{X\subset \mathcal{P}(\Lambda)} c_X \prod_{\mu\in X}\lim_{N\to\infty}\mathbb{E}_{\beta,h\bm{\epsilon}}\ s^\mu_i\nonumber\\
&=& \sum_{X\subset \mathcal{P}(\Lambda)} c_X \prod_{\mu\in X}\meanv{s^\mu}_{\bs,\bm{\epsilon}}=
\left\langle F(\bs) \right\rangle_{\bs,\bm{\epsilon}},
\end{eqnarray}
where in the second line we have used that in the Curie-Weiss model at $\beta>0$ and $h>0$
\begin{equation}
\lim_{N\to\infty} \mathbb{E}_{\beta,h\bm{\epsilon}}\ s^\mu_i= m_0(\beta,\epsilon^\mu h)=\meanv{s^\mu}_{\bs,\bm{\epsilon}}\ \ \forall i=1,\ldots,N.
\end{equation}

\end{proof}

\begin{lemma}\label{lemma:2}
Let $\bs_i\in\{-1,1\}^M$ the $i$-th marginal of $\bS=\{\bs_i\}_{i=1}^N\in\{-1,1\}^{NM}$, distributed according to (\ref{eq:pattern_d}). Given a compact set $K\subset\mathbb{R}^M$ and a bounded function $F:K\times\{-1,1\}^M\to \mathbb{R}$, it holds, uniformly in $K$, that 
\begin{equation}
   \frac 1 N \sum_{i=1}^N F(\bm{p},\bs_i)\overset{p}{\longrightarrow}\left\langle F(\bm{p},\bs) \right\rangle_{\bs,\bm{\epsilon}}.\nonumber
\end{equation}
\end{lemma}

\begin{proof}
Given the set 
\begin{equation}
A_{N,\epsilon}=\left\{\bS: \sup_{\bm{p}\in K} \left| \frac 1 N\sum_{i=1}^N F(\bm{p},\bs_i) - \left\langle F(\bm{p},\bs) \right\rangle_{\bs,\bm{\epsilon}}\right| >\epsilon\right\},
\end{equation}
we need to show that, $\forall \epsilon>0$, $\lim_{N\to\infty}\mathbb{P}(A_{N,\epsilon})= 0$. To this aim let's consider the Fourier decomposition of $F(\bm{p},\bs)$ as
\begin{equation}
F(\bm{p},\bs)=\sum_{X\in\mathcal{P}(\Lambda)} c_X(\bm{p}) \prod_{\mu\in X} s^\mu
\end{equation}
and define the set 
\begin{equation}
B_{N,\epsilon}=\left\{\bS: \sup_{X\in\mathcal{P}(\Lambda)} \left|\frac1 N\sum_{i=1}^N \prod_{\mu\in X} s^\mu_i -\meanv{\prod_{\mu\in X} s^\mu}_{\bs,\bm{\epsilon}}  \right|>\frac{\epsilon}{L_F} \right \},
\end{equation}
where $L_F=\sum_{X\in\mathcal{P}(\Lambda)}\sup_{\bm{p}\in K} |c_X(\bm{p})|<\infty$.
According to this definition note that within the set $B_{N,\epsilon}^c$ it holds
\begin{eqnarray}
\sup_{p\in K} \left| \frac 1 N\sum_{i=1}^N F(\bm{p},\bs_i) - \left\langle F(\bm{p},\bs) \right\rangle_{\bs,\bm{\epsilon}}\right|&\leq& 
\sup_{p\in K} \sum_{X\in\mathcal{P}(\Lambda)} |c_X(\bm{p})|  \left|\frac 1 N \sum_{i=1}^N \prod_{\mu\in X} s^\mu_i -\meanv{\prod_{\mu\in X} s^\mu}_{\bs,\bm{\epsilon}}  \right|\leq \epsilon. \nonumber 
\end{eqnarray}
Therefore  $B^c_{N\epsilon}\subseteq A^c_{N\epsilon}$ and thus $A_{N,\epsilon}\subseteq B_{N,\epsilon}$. As a consequence 
\begin{eqnarray}
\mathbb{P}(A_{N,\epsilon})&\leq& \mathbb{P}(B_{N,\epsilon})
= \mathbb{P}\left(\bigcup_{X\in\mathcal{P}(\Lambda)}   \left\{\bS:\left|\frac1 N\sum_{i=1}^N \prod_{\mu\in X} s^\mu_i -\meanv{\prod_{\mu\in X} s^\mu}_{\bs,\bm{\epsilon}}  \right|>\frac{\epsilon}{L_F} \right \}\right)
\nonumber \\
&\leq& \sum_{X\in\mathcal{P}(\Lambda)}\mathbb{P}\left(\left\{\bS:\left|\frac1 N\sum_{i=1}^N \prod_{\mu\in X} s^\mu_i -\meanv{\prod_{\mu\in X} s^\mu}_{\bs,\bm{\epsilon}}  \right|>\frac{\epsilon}{L_F} \right \}\right)\nonumber\\
&\leq& \sum_{X\in\mathcal{P}(\Lambda)}(L_F/\epsilon)^2 \ \mathbb{E}_{\beta,h\bm{\epsilon}}\left(\frac1 N\sum_{i=1}^N \prod_{\mu\in X} s^\mu_i -\meanv{\prod_{\mu\in X} s^\mu}_{\bs,\bm{\epsilon}} \right)^2 
\end{eqnarray}
that goes to zero in the thermodynamic limit since
\begin{equation}
   \mathbb{E}_{\beta,h\bm{\epsilon}}\  \left(\frac1 N\sum_{i=1}^N \prod_{\mu\in X} s^\mu_i \right)^2= \frac{1}{N^2}\sum_{i\neq j}\prod_{\mu\in X} \mathbb{E}_{\beta,h\bm{\epsilon}} \left(s^\mu_is^\mu_j\right) + \frac 1 N
\end{equation}
and using the factorization property of the Curie-Weiss model that 
\begin{equation}
    \lim_{N\to\infty} \mathbb{E}_{\beta,h\bm{\epsilon}} \left(s^\mu_is^\mu_j\right) = m^2_0(\beta,\epsilon^\mu h)=\meanv{s^\mu}^2_{\bs,\bm{\epsilon}}\ \ \forall i,j=1,\ldots,N, \ i\neq j.
\end{equation}
\end{proof}

\begin{lemma}
\label{lem:bound-delta.gn}The functions $g_N:\mathbb{R}^M\times\{-1,1\}^{NM}\to\mathbb{R}$ and $g:\mathbb{R}^M\to\mathbb{R}$ defined in eqs. (\ref{eq:magMat_sella}) and (\ref{eq:limsella})
satisfy
\begin{equation}
    \left|g_{N}(\bm{p},\bS)-g(\bm{p})\right|\le2\beta\sum_{\mu=1}^{M}\left|p_{\mu}\right|, \:\quad\forall \bm{p}\in\mathbb{R}^{M},\bS\in\{-1,+1\}^{NM}.
\end{equation}
As a consequence it holds
\begin{itemize}
    \item $\sup_N\sup_{K\times\{-1,1\}^{NM}} |g_N|<\infty$ for any compact $K\subset \mathbb{R}^M$;
    \item $\exists\  C_1<\infty$: $\sup_{\mathbb{R}^M\times\Sigma_N}g_N<C_1$ ;
    \item $\exists\  C_2<\infty$: $\int_{\mathbb{R}^M}e^{g_N(\bm{p},s)}d\bm{p}<C_2$;
    \item $\exists\  K\subset\mathbb{R}^M, \delta>0$: $ g_N(\bm{p},s)-\sup_K g <-\delta, \ \ \ \forall (\bm{p},s)\in K^c\times\Sigma_N$.
\end{itemize}
\end{lemma}

\begin{proof}
It holds, 
\begin{eqnarray*}
\left|g_{N}(\bm{p},\bS)-g(\bm{p})\right| & = & \left|\frac{1}{N}\sum_{i=1}^{N}\log\cosh\left(\beta\bs_{i}\cdot\bm{p}+\lambda\right)-\meanv{\log\cosh\left(\beta \bs\cdot \bm{p}+\lambda\right)}_{\bs,\epsilon}\right|\\
 & = & \left|\frac{1}{N}\sum_{i=1}^{N}\meanv{\log\frac{\cosh\left(\beta \bs_{i}\cdot \bm{p}+\lambda\right)}{\cosh\left(\beta\bm{\tau}\cdot \bm{p}+\lambda\right)}}_{\bm{\tau},\epsilon}\right|
  =  \left|\frac{1}{N}\sum_{i=1}^{N}\meanv{\int_{\beta\bm{\tau}\cdot \bm{p}+\lambda}^{\beta \bs_{i}\cdot \bm{p}+\lambda}\tanh(x)dx}_{\bm{\tau},\epsilon}\right|\\
  &\leq& \frac 1 N \sum_{i=1}^N \left|\meanv{\beta \bs_{i}\cdot \bm{p} -\beta\bm{\tau}\cdot \bm{p} }_{\bm{\tau},\bm{\epsilon}}\right|
\leq 2\beta \sum_{\mu=1}^M |p_\mu|.
\end{eqnarray*}
As a consequence of the previous relation $g_N$ is bounded uniformly by 
\begin{equation}
g_N(\bm{p},\bS)\leq g(\bm{p})+\left|g_N(\bm{p},\bS)-g(\bm{p})\right|
\leq
\log 2-\frac{\beta\bm{p}^2}{2} +\meanv{\log\cosh(\beta \bs\cdot\bm{p}+\lambda)}_{\bs,\bm{\epsilon}}+2\beta\sum_{\mu=1}^M|p_\mu|:= \hat{g}(\bm{p}).
\end{equation}
Since $\hat{g}$ is continuous, it is bounded in any compact set $K\subset\mathbb{R}^M$.  Moreover it is easy to see that $\sup_{\mathbb{R}^M}\hat{g}<\infty$ and that $\int_{\mathbb{R}^M}e^{g(\bm{p})}d\bm{p}<\infty$ since it has gaussian tails. Moreover since $\hat{g}$ goes to $-\infty$ at infinity it always exists a sufficiently large ball $B_{r_\delta}$ of radius $r_\delta$ such that $\hat{g}(\bm{p})-\sup \hat{g}<-\delta$ for $\bm{p}\in B_{r_\delta}^c$. Therefore the properties for $g_N$ are proved uniformly in $N$.
\end{proof}

\begin{proposition}\label{pr:sp1}
Let  $K\subset\mathbb{R}^M$ a compact set, $\mu_N$ a probability distribution over a finite set $\Sigma_N$, $F_N:K\times\Sigma_N\to \mathbb{R}$ a  sequence of bounded functions such that $\sup_N \sup_{K\times\Sigma_N}|F_N|<\infty$ and $F:K\to\mathbb{R}$ bounded  such that
\begin{itemize}
    \item $F_N \overset{p}{\longrightarrow} F$ uniformly in $K$.
    \item $\lim_{N\to\infty} \frac 1 N \log \int_K e^{NF} = \sup_K F$;
\end{itemize}
Then it holds
\begin{equation}
\lim_{N\to\infty} \frac 1 N \mathbb{E}_{\mu_N} \log \int_K d\bm{p} e^{NF_N(\bm{p},s)}=\sup_K F.
\end{equation}
\end{proposition}
\begin{proof}
Let's define $C=\sup_N \sup_{K\times\Sigma_N}|F_N|<\infty$ and the set
\begin{equation}
A_{N,\epsilon}=\left\{s\in\Sigma_N: \sup_{p\in K} \left|  F_N(\bm{p},s) - F(\bm{p}\right| >\epsilon\right\}.
\end{equation} By assumptions it holds, $\forall \epsilon>0$, that $\mu_N(A_{N,\epsilon})\to 0$. Therefore  
it holds, $\forall\epsilon>0 $, that
\begin{eqnarray}
\lim_{N\to\infty}\frac 1 N \mathbb{E}_{\mu_N} \mathbb{I}_{A_{N,\epsilon}}\log \int_K d\bm{p} e^{NF_N(\bm{p},s)}=0, \nonumber
\end{eqnarray}
where $\mathbb{I_A}$ is the indicator function of the set $A$, since, calling $|K|$ the Lebesgue measure of $K$, \begin{equation}
\left|\frac 1 N \mathbb{E}_{\mu_N} \mathbb{I}_{A_{N,\epsilon}}\log \int_K d\bm{p} e^{NF_N(\bm{p},s)}\right|\leq \mu_N(A_{N,\epsilon})\left|C+\log(|K|)/N\right| .
\end{equation}
Thanks to these results we just need to evaluate the expected value over $A^c_{N,\epsilon}$ that has full probability in the limit. Inside $A^c_{N,\epsilon}$ we can substitute $F_N$ with $F$ at the exponent with an error $\epsilon$ i.e. 
\begin{equation}
\left|\frac 1 N \mathbb{E}_{\mu_N} \mathbb{I}_{A^c_{N,\epsilon}}\log \int_K d\bm{p} e^{NF_N(\bm{p},s)}- \frac 1 N \mathbb{E}_{\mu_N} \mathbb{I}_{A^c_{N,\epsilon}}\log \int_K d\bm{p} e^{NF(\bm{p})}\right|<\epsilon.
\end{equation}
Since $\mathbb{E}_{\mu_N}\mathbb{I}_{A^c_{N,\epsilon}}=\mu_N(A^c_{N,\epsilon})\to 1$ and using  the standard Laplace approximation, it holds
\begin{equation}
\lim_{N\to\infty } \frac 1 N \mathbb{E}_{\mu_N} \mathbb{I}_{A^c_{N,\epsilon}}\log \int_K d\bm{p} e^{NF(\bm{p})} =\sup_K F,
\end{equation}
therefore, $\forall \epsilon >0$,
\begin{equation}
\left|\lim_{N\to\infty}\frac 1 N \mathbb{E}_{\mu_N} \log \int_K d\bm{p} e^{NF_N(\bm{p},s)}-\sup_K F\right|=\left|\lim_{N\to\infty}\frac 1 N \mathbb{E}_{\mu_N} \mathbb{I}_{A^c_{N,\epsilon}}\log \int_K d\bm{p} e^{NF_N(\bm{p},s)}-\sup_K F\right|<\epsilon.
\end{equation}

\end{proof}

\begin{proposition}\label{pr:sp2}
Let $K\subset\mathbb{R}^M$ a compact set, $\mu_N$ a probability distribution over a finite set $\Sigma_N$, $F_N:\mathbb{R}^M\times\Sigma_N\to \mathbb{R}$ and  $F:\mathbb{R}^M\to\mathbb{R}$ satisfying the conditions of Proposition \ref{pr:sp1} when restricted on the set $K$.
Assuming moreover that
\begin{itemize}
    \item $\exists\  C_1<\infty$: $\sup_{\mathbb{R}^M\times\Sigma_N}F_N<C_1$ ;
    \item $\exists\  C_2<\infty$: $\int_{\mathbb{R}^M}e^{F_N(\bm{p},s)}d\bm{p}<C_2$;
    \item $\exists\  \delta>0$: $ F_N(\bm{p},s)-\sup_K F <-\delta, \ \ \ \forall (\bm{p},s)\in K^c\times\Sigma_N$;
\end{itemize}
then it holds
\begin{equation}
\lim_{N\to\infty} \frac 1 N \mathbb{E}_{\mu_N} \log \int_{\mathbb{R}^M} d\bm{p} e^{NF_N(\bm{p},s)}=\sup_K F.
\end{equation}
\end{proposition}

\begin{proof}
Since
\begin{equation}
\frac 1 N \mathbb{E}_{\mu_N} \log \int_{\mathbb{R}^M} d\bm{p} e^{NF_N(\bm{p},s)}=
\frac 1 N \mathbb{E}_{\mu_N} \log \int_{K} d\bm{p} e^{NF_N(\bm{p},s)}+
\frac 1 N \mathbb{E}_{\mu_N} \log\left(1+\frac{ \int_{K^c} d\bm{p} e^{NF_N(\bm{p},s)}}{ \int_{K} d\bm{p} e^{NF_N(\bm{p},s)} }\right)
\end{equation}
it is sufficient to apply Proposition \ref{pr:sp1} for the first term and  showing that the second term is vanishing in the limit. Defining the set $A_{N,\epsilon}$ as in the proof of Proposition \ref{pr:sp1}, with $\epsilon<\delta$ and denoting $F^*=\sup_K F$ and $C_0=\sup_N\sup_{K\times\Sigma_N}|F_N|$, it holds
\begin{equation}
\frac 1 N \mathbb{E}_{\mu_N}\mathbb{I}_{A_{N,\epsilon}} \log\left(1+\frac{ \int_{K^c} d\bm{p} e^{NF_N(\bm{p},s)}}{ \int_{K} d\bm{p} e^{NF_N(\bm{p},s)} }\right)
\leq \mu_N(A_{N,\epsilon})\frac 1 N \log \left(1+ \frac{e^{(N-1)C_1}C_2}{|K|e^{-NC_0}}\right)\overset{N\to\infty}{\longrightarrow} 0,
\end{equation}
since $\mu_N(A_{N,\epsilon})\to 0$.
Moreover, choosing $\epsilon_2<\delta-\epsilon$, it holds $\forall N>N_{\epsilon_2}$
\begin{eqnarray}
&&\frac 1 N \mathbb{E}_{\mu_N}\mathbb{I}_{A^c_{N,\epsilon}} \log\left(1+\frac{ \int_{K^c} d\bm{p} e^{NF_N(\bm{p},s)}}{ \int_{K} d\bm{p} e^{NF_N(\bm{p},s)} }\right)
\leq 
\frac 1 N \mathbb{E}_{\mu_N}\mathbb{I}_{A^c_{N,\epsilon}} \frac{ \int_{K^c} d\bm{p} e^{NF_N(\bm{p},s)}}{ \int_{K} d\bm{p} e^{NF_N(\bm{p},s)} }\nonumber \\
&\leq&
\frac 1 N \mathbb{E}_{\mu_N}\mathbb{I}_{A^c_{N,\epsilon}} \frac{ e^{(N-1)F^*}\int_{K^c} d\bm{p} e^{(N-1) (F_N(\bm{p},s)-F^*)+F_N(\bm{p},s)}}{ \int_{K} d\bm{p} e^{-N|F_N(\bm{p},s)-F(\bm{p})|+NF(\bm{p})} }
\leq 
\frac 1 N  \frac{ e^{(N-1)F^*-\delta(N-1)}C_2}{ e^{-\epsilon N}\int_{K} d\bm{p} e^{NF(\bm{p})} }
\nonumber \\
&\leq&
\frac 1 N  \frac{ e^{(N-1)F^*-\delta(N-1)}C_2}{ e^{-\epsilon N} e^{NF^*-N\epsilon_2} }=C e^{-N(\delta-\epsilon-\epsilon_1)}\overset{N\to\infty}{\longrightarrow} 0,
\nonumber
\end{eqnarray}
where we have used that 
\begin{equation}
\left|\frac 1 N \log \int_{K} d\bm{p} e^{NF(\bm{p})}-F^* \right|<\epsilon_2.
\end{equation}
\end{proof}

\begin{proof}{of Proposition \ref{prop:1}}\\
\noindent To find the value of the magnetization (third point) we use that, $\forall N\in\mathbb{N}$,
\begin{equation}
  \meanv{\frac{\bm{1}\cdot\bxi}{N}}=-\beta\partial_\lambda\  f_N .
\end{equation}
Thanks to the concavity of $f_N$ in $\lambda$, we can exchange the thermodynamic limit with the derivative  obtaining 
\begin{equation}
    m=
    \partial_\lambda (-\beta f)= \left\langle  \tanh(\beta \bm{s} \cdot \bm{p} +\lambda)\right\rangle_{\bm{s}}.\nonumber 
\end{equation}
Moreover
\begin{equation}
 \operatorname{Var}(\frac{\bm{1}\cdot\bxi}{N})=-\beta N^{-1} \ \partial^2_\lambda f_N
\end{equation} 
and therefore the magnetization has to be self-averaging in the thermodynamic limit. 
Analogous arguments, based on the response of the free energy to linear external perturbations \cite{mezard2009information} can be used for the first two points that are just generalizations of classical results about the Curie Weiss model \cite{ellis2006entropy}.

\end{proof}
\begin{proof}{of Proposition \ref{prop:p}}\\
\noindent  Applying the Fourier decomposition to the function $f(\bs;\bm p)=\bs\tanh(\beta \bs\cdot\bm{p})$ it holds
\begin{eqnarray}\label{eq:fourier}
p^\mu &=& \meanv{f^\mu(\bs;\bm p)}_{\bs} = A^\mu(\bm p) + \sum_{\nu\neq \mu} A^{\nu}(\bm p) \left\langle s_\mu s_\nu \right\rangle_{\bm s} = A^\mu(\bm p) + m_0^2(\beta) \sum_{\nu\neq \mu} A^{\nu}(\bm p) 
\end{eqnarray}
where
\begin{equation}
A^\mu(\bm p) = \langle\tanh(\beta p^\mu + \beta\sum_{\nu\neq \mu} p^\nu s_\nu) \rangle_{\bm{s},0}
\end{equation}
In fact it is easy to check by symmetry that
\begin{eqnarray}
\meanv{s_\mu\tanh(\beta(\bm{s}\cdot \bm{p})),1}&=& A^\mu(\bm p) \nonumber\\
\meanv{s_\mu\tanh(\beta(\bm{s}\cdot \bm{p})),s_\nu}&=& \meanv{\tanh(\beta (p^\nu s_\mu + p^\mu s_\nu +\sum_{k\neq \mu,\nu} p^k s_k s_\mu s_\nu ))}_{\bm s, 0}=0 \ \ \ \forall \nu\in \Lambda, \nu\neq \mu \nonumber \\
\meanv{s_\mu\tanh(\beta(\bm{s}\cdot \bm{p})),s_\mu s_\nu}&=& \meanv{s_\nu\tanh(\beta(\bm{s}\cdot \bm{p})),1} =A^\nu(\bm p)\ \ \forall \nu\in \Lambda, \nu\neq \mu \nonumber \\
\meanv{s_\mu\tanh(\beta(\bm{s}\cdot \bm{p})),s_k s_\nu}&=& \meanv{\tanh(\beta(p^\mu s_k s_\nu + \sum_{l\neq \mu} p^l s_l s_\mu s_k s_\nu))}_{\bm s, 0}= 0\ \ \ \forall k,\nu \in \Lambda,  k,\nu \neq \mu\nonumber \\
\meanv{s_\mu\tanh(\beta(\bm{s}\cdot \bm{p})),s_X}&=& 0  \ \ \ \forall X\subset\mathcal{P}(\Lambda), |X|>2\nonumber.
\end{eqnarray}
Using eq. (\ref{eq:fourier}) it holds $\forall  \mu,\nu \in \Lambda, \mu\neq\nu$, that
\begin{equation}
p^\mu-p^\nu =(1-m_0^2(\beta))\left(A^\mu(\bm p)-A^\nu(\bm p)\right).
\end{equation}
and by direct computation we have that
\begin{eqnarray}
A^\mu(\bm p)-A^\nu(\bm p)
&=& \meanv{\tanh(\beta p^\mu + \beta\sum_{k\neq \mu} p^k s_k) }_{\bm{s},0}
- \meanv{\tanh(\beta p^\nu + \beta\sum_{l\neq \nu} p^l s_l) }_{\bm{s},0}\nonumber \\
&=&\frac 1 2\left( \meanv{\tanh(\beta p^\mu + \beta p^\nu+ \beta\sum_{k\neq \mu,\nu} p^k s_k) }_{\bm{s},0} + \meanv{\tanh(\beta p^\mu - \beta p^\nu+ \beta\sum_{k\neq \mu,\nu} p^k s_k) }_{\bm{s},0}\right)\nonumber\\
&-&\frac 1 2\left( \meanv{\tanh(\beta p^\nu + \beta p^\mu+ \beta\sum_{k\neq \mu,\nu} p^k s_k) }_{\bm{s},0} + \meanv{\tanh(\beta p^\nu - \beta p^\mu+ \beta\sum_{k\neq \mu,\nu} p^k s_k) }_{\bm{s},0}\right)\nonumber\\
&=&  \meanv{\tanh(\beta (p^\mu - p^\nu)+ \beta\sum_{k\neq \mu,\nu} p^k s_k) }_{\bm{s},0}\nonumber.
\end{eqnarray}
Thus $p^\mu-p^\nu$ satisfies an equation of the form
\begin{equation}
p^\mu-p^\nu=(1-m_0^2)\meanv{\tanh(\beta (p^\mu - p^\nu)+ \beta Z^{\bm p}) }_{Z^{\bm p}}
\end{equation}
with $Z^{\bm p}=\sum_{k\neq \mu,\nu} p^k s_k$ a random noise. The function $A \meanv{\tanh(B x + Z^{\bm p})}_{Z^{\bm p}}$ is odd and for $x\geq 0$ it is always under the line $AB x$. In fact, for any vector $\bm p$ and $\lambda\in(0,1)$, 
\begin{eqnarray}
\frac{d}{d\lambda} A \meanv{\tanh(B x + \lambda \bm p \cdot \bm s)}_{\bm s,0}&=& 
A\meanv{ ( \bm p \cdot \bm s)\left( 1-\tanh^2(B x + \lambda \bm p \cdot \bm s)\right)}_{\bm s,0}\nonumber \\
&=& - A\meanv{ ( \bm p \cdot \bm s)\tanh^2(B x + \lambda \bm p \cdot \bm s)}_{\bm s,0}\nonumber\\
&=& - \frac{A}{2}\meanv{ ( \bm p \cdot \bm s)\left(\tanh^2(B x + \lambda \bm p \cdot \bm s)-
\tanh^2(B x - \lambda \bm p \cdot \bm s)\right)  }_{\bm s,0}\nonumber \\
&\leq& 0
\end{eqnarray}
since $\tanh^2(x+y)-\tanh^2(x-y) \geq 0$, $\forall x,y \geq 0$. As a consequence it holds $\forall \bm p$ and $x\geq 0$
\begin{equation}
A \meanv{\tanh(B x + \bm p \cdot \bm s)}_{\bm s,0}\leq A \tanh(B x) \leq AB x.
\end{equation}
Therefore we have that 
\begin{equation}\label{eq:dim}
\left|p^\mu - p^\nu\right| \leq \beta(1-m_0^2(\beta)) \left|p^\mu - p^\nu\right|. 
\end{equation}
From the theory of the Curie Weiss model it holds that $\beta(1-m_0^2(\beta))<1$ at any temperature. In fact if $\beta<1$ it is $m_0(\beta)=0$ and $\beta(1-m_0^2(\beta))=\beta<1$. On the contrary if $\beta>1$,  $\tanh(\beta x)$ intersects the bisector from above at the point $x=m_0(\beta)>0$, thus with
\begin{equation}
1>\frac{d}{dx}\tanh(\beta x)|_{m_0}=\beta(1-\tanh^2(\beta m_0))=\beta(1-m_0^2).
\end{equation}
Therefore the only solution of (\ref{eq:dim}) is $\left|p^\mu - p^\nu\right| =0$.
\end{proof}


\begin{proof}{of Proposition \ref{prop:2}}\\
\noindent We have to solve
\begin{equation}
\bar{p}=\sum_{\bm{s}} s_1 \frac{e^{\beta m_0\sum_{\mu=1}^Ms^\mu}}{\left(2\cosh(\beta m_0) \right)^M} \tanh (\beta \bar{p} \sum_{\mu=1}^M s^\mu).
\end{equation}
By symmetry ($\bm{s}\to-\bm{s}$) this is equivalent to 
\begin{equation}
\bar{p}=\sum_{\bm{s}} s_1 \frac{\cosh(\beta m_0\sum_{\mu=1}^Ms^\mu)}{\left(2\cosh(\beta m_0) \right)^M} \tanh (\beta \bar{p} \sum_{\mu=1}^M s^\mu).
\end{equation}
Evaluating the rhs in the point $\bar{p}=m_0$ and denoting $z=2\cosh(\beta m_0)$ we have
\begin{eqnarray}
&& \sum_{\bm{s}} s_1 \frac{\sinh(\beta m_0\sum_{\mu=1}^Ms^\mu)}{\left(2\cosh(\beta m_0) \right)^M} = z^{-M} \sum_{\bm{s}} s_1 \sinh(\beta m_0\sum_{\mu=1}^{M-1}s^\mu + \beta m_0 s^M)\nonumber \\
&=& z^{-M} \sum_{s^M}\sum_{s^1,\ldots,s^{M-1}} s_1 \left[\sinh(\beta m_0\sum_{\mu=1}^{M-1}s^\mu )\cosh(\beta m_0 s^M) + \cosh(\beta m_0\sum_{\mu=1}^{M-1}s^\mu )\sinh(\beta m_0 s^M) \right]\nonumber\\
&=& z^{-M} \sum_{s^M}\sum_{s^1,\ldots,s^{M-1}} s_1 \sinh(\beta m_0\sum_{\mu=1}^{M-1}s^\mu )\cosh(\beta m_0 s^M) = \sum_{s^1,\ldots,s^{M-1}}s_1\frac{ \sinh(\beta m_0\sum_{\mu=1}^{M-1}s^\mu )}{\left(2\cosh(\beta m_0) \right)^{M-1}}\nonumber\\
&=&\ldots = \sum_{s^1}s_1\frac{ \sinh(\beta m_0 s^1 )}{2\cosh(\beta m_0)} =\tanh(\beta m_0). \nonumber
\end{eqnarray}
\end{proof}


\begin{proof}{of Proposition \ref{prop:3}}\\
Thanks to Proposition \ref{prop:p} it is sufficient to evaluate the free energy along the line $p^\mu=\bar{p}$, where
\begin{equation}
f(\bar{p})= \frac{M}{2} \bar{p}^2 -\frac{1}{\beta}\meanv{\log 2\cosh (\beta\bar{p} \sum_{\mu=1}^M s_\mu )}_{\bm s} .
\end{equation}
Since $f(\bar{p})$ is an even function we can just study the branch $\bar{p}\geq 0$. As a function of $\bar{p}^2$ the free energy is convex since
\begin{equation}
\frac{df}{d\bar{p}^2}=\frac{f'(\bar{p})}{2\bar{p}}=\frac{M}{2}\left(1-\frac{\meanv{s_1\tanh(\beta\bar{p} \sum_{\mu=1}^M s_\mu )}_{\bm s}}{ \bar{p}} \right)
\end{equation}
is an increasing function, therefore the position of the minimum depends on the sign of the derivative in zero. Since for $\beta>1$ 
\begin{equation}
\frac{df}{d\bar{p}^2}|_{\bar{p}=0}=\frac M 2 \left( 1 - \beta -\beta(M-1)m_0^2 \right)<\frac M 2 (1-\beta) <0,
\end{equation}
the minimum is attained away from zero.
\end{proof}

\section*{Acknowledgement}
FA, DT and GM acknowledge GNFM-Indam for financial support. This work was partially supported by project SERICS (PE00000014) under the MUR National Recovery and Resilience Plan funded by the European Union - NextGenerationEU. 

\bibliography{refs}
\bibliographystyle{unsrt}


\appendix

\section{Replica computation of the RS conjectures}\label{appendix:replica}

We define the model partition function as
\begin{equation}
    Z(\bm{\mathcal{S}}) := \sum_{\bm{\xi}}\exp \left( \frac{\beta}{N} \sum_{b=1}^M \sum_{i<j} s_i^b s_j^b \xi_i \xi_j \right)=\sum_{\bm{\xi}}\exp \left( \frac{\beta}{2N} \sum_{b=1}^M \sum_{i,j} s_i^b s_j^b \xi_i \xi_j - M\frac{\beta}{2}\right).
\end{equation}
We assume that the system could be aligned with a subset  $\ell_1 = \mathcal{O}(1) $ of examples, extracted by the Curie-Weiss at inverse temperature $\hat{\beta}$ and we measure the corresponding overlaps with
\begin{equation}
    p^{b}(\bxi) = \frac{1}{N} \sum_{i=1}^N s_i^b \xi_i \quad \quad b=1,...,\ell_1 \;.
\end{equation}
to get 
\begin{equation}
    Z(\bm{\mathcal{S}}) =\sum_{\bm{\xi}}\exp \left( 
    \frac{\beta N}{2} \sum_{b=1}^{\ell_1}
    (p^b(\xi) )^2 +
    \frac{\beta}{2N} \sum_{b=\ell_1+1}^M \sum_{i,j} s_i^b s_j^b \xi_i \xi_j 
    \right) \;.
\end{equation}
In the present analysis we focus on the case when $\ell_1 =1$ for the sake of brevity.
Our aim is to compute the disorder average free energy density:
\begin{equation}
    -\beta f(\beta,\hat{\beta},\gamma) = \lim_{N \to \infty} \frac{1}{N} \left[ \ln Z \right]^{\bm{\mathcal{S}}}\,,
\end{equation}
here $\left[\cdots\right]^{\bm{\mathcal{S}}}$ corresponds to the average versus disorder, given by the dual patterns.
It is possible to rewrite the previous expression by exploiting the standard replica trick as
\begin{equation}\label{G:Replica Trick2}
-\beta f(\beta,\hat{\beta},\gamma)  = \lim_{\substack{N \to \infty \\ n\to 0}}\frac{ \ln [Z^n]^{\bm{\mathcal{S}}}}{Nn} 
\end{equation}
where
\begin{align*}
        \left[Z^n\right]^{\bm{\mathcal{S}}} & = \sum_{\bm{\mathcal{S}}} P^{CW}_{\hat{\beta}}(\bm{\mathcal{S}})\ Z^n(\bm{\mathcal{S}})
        = \sum_{\bm{\mathcal{S}}}\ \prod_{b=1}^M \frac{1}{z(\hat{\beta})}\  e^{\hat{\beta}/2N \sum_{i,j}s_i^bs_j^b} \ Z^n(\bm{\mathcal{S}}) \,,
\end{align*}
and
\begin{align*}
        Z^n(\bm{\mathcal{S}})
        = \sum_{\bm{\xi}^1,\ldots,\bm{\xi}^n}   \exp \left( 
    \frac{\beta N}{2} \sum_{a=1}^n
    (p^{a}(\xi))^2 +
    \frac{\beta}{2N} \sum_{a=1}^n\sum_{b=2}^M \sum_{i,j} s_i^b s_j^b \xi_i^a \xi_j^a 
    \right)\, .
\end{align*}
The summation over $\bm{\mathcal{S}}$, is intended over $(\bm{s}^1,\ldots,\bm{s}^M)$.
One can subdivide $\left[Z^n\right]^{\bm{\mathcal{S}}}$ in two parts: the one related to the first aligned example (say $\bm{s^1}$) and the second related to the others. These two pieces will be evaluated separately
\begin{align*}
        \left[Z^n\right]^{\bm{\mathcal{S}}} &
        = \sum_{\bm{\xi}^1,\ldots,\bm{\xi}^n}  \sum_{\bm{s}}\ \frac{1}{z(\hat{\beta})}\  \exp \left( \frac{\hat{\beta}}{2N} \sum_{i,j}s_i s_j +
    \frac{\beta N}{2} \sum_{a}
    (p^{a}(\xi))^2 \right) \times \\
    & 
    \phantom{\sum_{\bm{\xi}^1,\ldots,\bm{\xi}^n \sum \frac{1}{Z}}  } 
    \times\sum_{\bm{\mathcal{S}}} \prod_{b=2}^M \frac{1}{z(\hat{\beta})} \exp \left( \frac{\hat{\beta}}{2N} \sum_{i,j}s_i^b s_j^b + \frac{\beta}{2N} \sum_{a}\sum_{i,j}s_i^b s_j^b\xi_i^a\xi_j^a \right) \,,\\
    &
        = \sum_{\bm{\xi}^1,\ldots,\bm{\xi}^n}  \sum_{\bm{s}}\ \frac{1}{z(\hat{\beta})}\  \exp \left( \frac{\hat{\beta}N}{2} m_0^2(s) +
    \frac{\beta N}{2} \sum_{a}
    (p^{a}(\xi))^2 \right) \times \\
    & 
    \phantom{\sum_{\bm{\xi}^1,\ldots,\bm{\xi}^n \sum \frac{1}{Z}}    } \times\left(\frac{2^N}{z(\hat{\beta})} \overline{e^{\hat{\beta}/2N \sum_{i,j}s_i s_j + \beta/2N\sum_a\sum_{i,j}s_i s_j\xi_i^a\xi_j^a}}^{\ \bm{s}}\right)^{M-1} \,,
\end{align*}
where in the first line we simply drop out the index for the first example and we introduced the magnetization as $m_0(\bs) = 1/N \sum_i^N s_i$, while in the second line we denoted by $\overline{\cdots}^{\ \bm{s}}$ the average w.r.t a single example. The second term can be computed by introducing  Gaussian variables $(z^1,\ldots,z^n)$ and $z$ to linearize the exponent as
\begin{align*}
       \phantom{=}& \int\prod_{a=1}^n Dz^a Dz\  \overline{e^{\sqrt{\beta/N}\sum_i \sum_a z^a \xi_i^as_i + \sqrt{\hat{\beta}/N}\sum_i z s_i}}^{ \ \bm{s}}  =\int\prod_{a=1}^n Dz^a Dz\ e^{\sum_i \ln \cosh(\sqrt{\beta/N}\sum_a z^a \xi_i^a + \sqrt{\hat{\beta}/N} z)}
\end{align*}
\begin{equation}
    \approx \int\prod_{a=1}^n Dz^a Dz\ e^{\beta/2 \left[\sum_{a\neq b} z_a z_b q_{ab} + \sum_a z_a^2 \right] + \hat{\beta}/2 z^2 +\sqrt{\hat{\beta} \beta} z \sum_a z_a m_a } =: \det\big(\bm{\Xi}(\bm{q},\bm{m})\big)^{-1/2}\,,
\end{equation}
where in the last line we have expanded the $\ln \cosh(x)$ and we have introduced the quantities
\begin{equation}
    q_{ab}(\bm{\xi})=\frac{1}{N}\sum_{i=1}^N \xi^a_i\xi^b_i \,,\ \ \ \ \ m_a(\bm{\xi})=\frac{1}{N}\sum_{i=1}^N \xi^a_i \;.
\end{equation}
The averaged partition function thus becomes
\begin{equation}
        \left[Z^n\right]^{\bm{\mathcal{S}}} 
        = 
         \sum_{\bm{\xi}^1,\ldots,\bm{\xi}^n} \sum_{\bm{s}}\ e^{  \frac{\hat{\beta}N}{2} m_0(\bm{s})^2 +    \frac{\beta N}{2} \sum_{a}
    (p^{a}(\bm{\xi}))^2 } 
    \  \frac{2^{N(M-1)}}{z(\hat{\beta})^{M}}
       \det\big(\bm{\Xi}(\bm{q},\bm{m})\big)^{(1-M)/2} .
\end{equation}
If we denote by $\mathcal{D}(m_0)$ the density of states for the $\bm{s}$ configuration
\begin{align*}
    \mathcal{D}(m_0) = & \sum_{\bm{s}} \delta\left( m_0 - \frac{1}{N}\sum_i s_i\right)\propto \sum_{\bm{s}}\int d\hat{m}_0 \exp{-N \hat{m}_0 \left( m_0 -\frac{1}{N}\sum_i s_i \right)}
\end{align*}
and with $\mathcal{D}(\bm{q}, \bm{m}, \bm{p})$ the one related to the states of the student machine
\begin{align*}
\small
    \mathcal{D}(\bm{q}, \bm{m}, \bm{p}) & = \sum_{\bm{\xi}^1,\ldots,\bm{\xi}^n} \prod_{a<b}\delta\left(q_{ab}-\frac{1}{N}\sum_i \xi_i^a \xi_i^b\right)\prod_{a}\delta \left( m_a-\frac{1}{N}\sum_i\xi_i^a \right) \prod_{a}\delta \left( p_a-\frac{1}{N}\sum_i s_i \xi_i^a \right)\\
    & \propto  \sum_{\bm{\xi}^1,\ldots,\bm{\xi}^n} \int \prod_{a<b} d\hat{q}_{ab}\   \exp\left(-N\sum_{a<b} \hat{q}_{ab} \left(q_{ab}-\frac{1}{N}\sum_i\xi^a_i\xi^b_i \right) \right) \times \\
    &\times \int \prod_{a}  d\hat{m_a}\ \exp\left( -N\sum_a\hat{m}_a\left(m_a- \frac{1}{N}\sum_i \xi^a_i \right)\right) \int \prod_{a} d\hat{p}_a\   \exp\left(-N\sum_a\hat{p}_a\left(p_a- \frac{1}{N}\sum_i s_i \xi^a_i \right)\right) 
\end{align*}
the averaged partition function becomes
\begin{align*}
        \left[Z^n\right]^{\bm{\mathcal{S}}} &
        \propto
         \int d m_0 \  \mathcal{D}(m_0) \int d \bm{q} d \bm{m} d \bm{p} \ \mathcal{D}(\bm{q}, \bm{m}, \bm{p})  \ e^{  \frac{\hat{\beta}N}{2} m^2_0 +
    \frac{\beta N}{2} \sum_{a}
    (p^{a})^2 } \  \frac{2^{N(M-1)}}{z(\hat{\beta})^{M}}
       \det\big(\bm{\Xi}(\bm{q},\bm{m})\big)^{(1-M)/2} \,,
\end{align*}
and then by saddle point approximation the free energy density results as
\begin{equation}
-\beta f 
 \approx  \lim_{\substack{N \to \infty \\ n\to 0}} \frac{1}{N n} \ln \Bigg( e^{N\extr f(\bm{q},\bm{m},\bm{p},\bm{m_0};\bm{\hat{q}},\bm{\hat{m}},\bm{\hat{p}},\bm{\hat{m}_0})} \Bigg) \approx  \lim_{n\to 0} \frac{1}{ n} \extr f(\bm{q},\bm{m},\bm{p},\bm{m_0};\bm{\hat{q}},\bm{\hat{m}},\bm{\hat{p}},\bm{\hat{m}_0})
\end{equation}
where
\begin{align*}
f(\bm{q},\bm{m},\bm{p},\bm{m_0};\bm{\hat{q}},\bm{\hat{m}},\bm{\hat{p}},\bm{\hat{m}_0}) & =  \frac{\beta}{2} \sum_{a}
    (p^{a})^2 + \frac{\hat{\beta}}{2} m^2_0 + (M-1) \ln 2 - M f_{CW} -\frac{\gamma}{2} \ln \det\big(\bm{\Xi}(\bm{q},\bm{m})\big) \\
    &  -\hat{m}_0 m_0 -\sum_a \hat{p}^a p^a -\sum_{a<b}\hat{q}^{ab} q^{ab}-\sum_a \hat{m}^a m^a \\
    &+ \ln \sum_{s} e^{\hat{m}_0 s} \sum_{\xi^1,\ldots,\xi^n} \exp \left( \sum_a \hat{p}^a\ s\  \xi_a + \sum_{a<b}\hat{q}^{ab}\xi_a \xi_b +\sum_a \hat{m}^a \xi_a\right)  \,.    
\end{align*}
The Replica symmetric (RS) ansatz assumes  $q_{ab}=q, m_a=m, p_a = p,\ \forall a,b=1,\ldots,n$.
Under the RS ansatz the term $\ln \det\big(\bm{\Xi}(\bm{q},\bm{m})\big)$ can be computed as
\begin{align*}
    \det\big(\bm{\Xi}(q,m)\big)^{-1/2}  &:= \int\prod_{a=1}^n Dz_a Dz \exp\left( \frac{\beta q}{2} \sum_{a\neq b}z_a z_b + \frac{\beta}{2}\sum_a z_a^2 + \frac{\hat{\beta} z^2}{2} + \sqrt{\hat{\beta} \beta } m\sum_a z_a z \right) \\
    & = \int\prod_{a=1}^n Dz_a \exp\left( \frac{\beta q}{2} \sum_{a\neq b}z_a z_b + \frac{\beta}{2}\sum_a z_a^2 \right) \int\frac{dz}{\sqrt{2\pi}}\exp\left(-(1-\hat{\beta}) \frac{z^2}{2} + \sqrt{\hat{\beta} \beta } m\sum_a z_a z \right) \\
    & = \int\prod_{a=1}^n Dz_a \exp\left( \frac{\beta q}{2} \sum_{a\neq b}z_a z_b + \frac{\beta}{2}\sum_a z_a^2 \right) (1-\hat{\beta})^{-\frac{1}{2}} \exp\left(\frac{\hat{\beta} \beta  m^2}{2(1-\hat{\beta})}(\sum_{a\neq b} z_a z_b + \sum_a z_a^2) \right) \\
    & = (1-\hat{\beta})^{-\frac{1}{2}} \det\left(\mathbb{I}-\beta\bm{q}-\frac{\hat{\beta} \beta  m^2}{2(1-\hat{\beta})}\bm{1}\right)^{-\frac{1}{2}} \,,
\end{align*}
where we introduced the matrix  $\bm{q}=(1-q)\bm{I}+q\bm{1}$. A matrix of the form $A\bm{I}+B\bm{1}$ has eigenvalues $\lambda_1=A+nB$ with multiplicity $1$ and $\lambda_2=A$ with multiplicity $n-1$. Hence $\ln \det\big(\bm{\Xi}(\bm{q},\bm{m})\big)$ can be expressed in the $n \to 0$ limit as
\begin{align*}
    \ln \det\big(\bm{\Xi}(\bm{q},\bm{m})\big) & = \ln \left( (1-\hat{\beta}) \det\left(\mathbb{I}-\beta\bm{q}-\frac{\hat{\beta} \beta  m^2}{2(1-\hat{\beta})}\bm{1}\right) \right)\\
    & = \ln \left[  (1-\hat{\beta})(1- \beta +\beta q) -n ((1-\hat{\beta})\beta q  +\hat{\beta}\beta m^2) \right] + (n-1)\ln\left[(1-\hat{\beta})(1- \beta +\beta q) \right]\\
    & \approx n \left[ \ln\left((1-\hat{\beta})(1- \beta +\beta q) \right) -\frac{\beta(1-\hat{\beta}) q +\hat{\beta}\beta m^2}{(1-\hat{\beta})(1- \beta +\beta q)}\right].
\end{align*}
Finally the term related to the density of states becomes under the RS ansatz
\begin{align*}
       &\sum_s \ e^{\hat{m}_0 s-n\frac{\hat{q}}{2}}  \int \mathcal{D}z \left[\sum_{\xi} \exp \left( \xi (  \hat{p}\ s\   +  \sqrt{\hat{q}}z+\hat{m})\right) \right]^n = 2\mathbb{E}_s \ e^{\hat{m}_0 s-n\frac{\hat{q}}{2}}  \int \mathcal{D}z \left[  2\mathbb{E}_{\xi} \exp \left( \xi (  \hat{p}\ s\   +  \sqrt{\hat{q}}z+\hat{m})\right) \right]^n  \;,\\
      &\approx e^{-n\frac{\hat{q}}{2}}\left( \ 2 \mathbb{E}_s \ e^{\hat{m}_0 s} + n \ 2 \mathbb{E}_s \ e^{\hat{m}_0 s}  \int \mathcal{D}z \   \ln 2  \mathbb{E}_{\xi} \ e^{ \xi (  \hat{p}\ s\   +  \sqrt{\hat{q}}z+\hat{m})}  \right)\;\\
      &\approx e^{-n\frac{\hat{q}}{2}} \ 2 \mathbb{E}_se^{\hat{m}_0 s} \left(1 + \frac{n}{\mathbb{E}_se^{\hat{m}_0 s}}  \mathbb{E}_s \ e^{\hat{m}_0 s}  \int \mathcal{D}z \   \ln 2  \mathbb{E}_{\xi} \ e^{ \xi (  \hat{p}\ s\   +  \sqrt{\hat{q}}z+\hat{m})}  \right)\;.
\end{align*}
Putting all together we define
\begin{equation}
-\beta f^{RS} (\beta,\hat{\beta},\gamma)
 =  \lim_{ n\to 0} \frac{1}{ n} \extr f^{RS}(p,q,m;\hat{p},\hat{q},\hat{m}) \,.
\end{equation}
where
\begin{align*}
f^{RS}(p,q,m;\hat{p},\hat{q},\hat{m})  &= \ln2 + \ln \cosh\hat{m}_0  + \frac{\hat{\beta}}{2} m^2_0 + (M-1) \ln 2 - M f_{CW}  -\hat{m}_0 m_0 + n\frac{\beta}{2}
    p^2 -n \hat{p} p+\\
    &   +\frac{n}{2}\hat{q}q-n\hat{m} m -n\frac{\gamma}{2}\ln\left((1-\hat{\beta})(1- \beta +\beta q) \right) +n\frac{\gamma}{2}  \frac{\beta(1-\hat{\beta}) q +\hat{\beta}\beta m^2}{(1-\hat{\beta})(1- \beta +\beta q)}+ \\
    &-n \frac{\hat{q}}{2}  + n \ln2 +\frac{n}{\mathbb{E}_s e^{\hat{m}_0 s}} \  \mathbb{E}_s e^{\hat{m}_0 s} \int \mathcal{D}z \ln \cosh \left(\hat{p}s + z\sqrt{\hat{q}} + \hat{m}\right) + \mathcal{O}(n^2)  \,. 
\end{align*}
Since $\hat{\beta}<1$, by extremizing with respect to $(\hat{m_0},m_0)$ we obtain $m_0 = 0$ and therefore
\begin{align*}
-\beta f^{RS}(\beta,\hat{\beta},\gamma) 
 = &\extr \Bigg[ -\frac{\gamma}{2} \ln\left((1-\hat{\beta})(1- \beta +\beta q) \right) +\frac{\gamma}{2}\frac{\beta(1-\hat{\beta}) q +\hat{\beta}\beta m^2}{(1-\hat{\beta})(1- \beta +\beta q)}  +\frac{1}{2}\hat{q}q  +  \\
    & -\hat{m}m - \frac{\hat{q}}{2}-\hat{p}p+ \frac{\beta}{2}p^2+  \ln2 + \mathbb{E}_s \int \mathcal{D}z \ln \cosh \left(\hat{p} s + z\sqrt{\hat{q}} + \hat{m}\right)\Bigg].
\end{align*}
 The equations for the saddle point are 
 \begin{align*}
    m&=\mathbb{E}_s\int Dz \tanh(\hat{m}+z\sqrt{\hat{q}} +\hat{p}s)\\ \hat{m}&= \frac{\gamma \hat{\beta}\beta m}{(1-\hat{\beta})(1-\beta+\beta q)}\\
    q &= \mathbb{E}_s \int Dz \tanh^2(\hat{m}+z\sqrt{\hat{q}} +\hat{p}s)\\
    \hat{q} &=  \frac{\gamma \hat{\beta} \beta^2(m^2)}{(1-\hat{\beta})(1-\beta+\beta q)^2}+\frac{\gamma \beta^2 q}{(1-\beta+\beta q)^2}\\ 
    p&= \mathbb{E}_s\int Dz \tanh(\hat{m}+z\sqrt{\hat{q}} +\hat{p}s) s\\
    \hat{p} &= \beta p
\end{align*}
where integration by parts was used to calculate $q$.
In particular on the Nishimori line, where $\hat{\beta}=\beta <1$ and $p=0$, it is
\begin{equation}
    -\beta f^{RS}(\beta,\gamma) = \extr \Bigg[ -\frac{\gamma}{2} \ln \det \bm{\Xi}\Bigr|_{\hat{\beta}=\beta} + \frac{1}{2}\hat{q}q - \hat{m}m - \frac{\hat{q}}{2}+  \ln 2 + \int Dz\ln \cosh(\hat{m}+z\sqrt{\hat{q}}) \Bigg]
\end{equation}
and 
\begin{align*}
   m&=\int Dz \tanh(\hat{m}+z\sqrt{\hat{q}})\\ \hat{m}&= \frac{\gamma \beta^2 m}{(1-\beta)(1-\beta+\beta q)}\\
    q &= \int Dz \tanh^2(\hat{m}+z\sqrt{\hat{q}})\\
    \hat{q} &=  \frac{\gamma \beta^3(m^2-q^2)}{(1-\beta)(1-\beta+\beta q)^2}+\frac{\gamma \beta^2 q}{(1-\beta)(1-\beta+\beta q)} \,.
\end{align*}
\end{document}